\documentclass[12pt]{article}
\usepackage{fullpage}
\usepackage{amssymb, amsthm, amsmath}
\usepackage{graphicx}
\usepackage[authoryear]{natbib}
\usepackage{bm}
\usepackage{times}
\usepackage{multirow}
\usepackage{float}
\usepackage{bbm}
\usepackage{mathtools}
\usepackage{setspace}
\usepackage{comment}

\usepackage[linesnumbered,ruled,vlined]{algorithm2e}
\usepackage{url}
\usepackage{enumitem}

\usepackage{caption}
\newcommand{\calD}{{\cal D}}

\newcommand{\calU}{{\cal U}}
\newcommand{\calI}{{\cal I}}

\newcommand{\RR}{\mathbb{R}}

\newcommand{\rank}{\mathrm{rank}}

\newcommand{\Matern}{ \mbox{Mat$\acute{\mbox{e}}$rn}}

\newcommand{\prob}{\mathsf{P}}
\newcommand{\E}{\mathsf{E}}
\newcommand{\var}{\mathsf{V}}

\newcommand{\beq}{ \begin{equation}}
\newcommand{\eeq}{ \end{equation}}
\newcommand{\beqn}{ \begin{eqnarray}}
\newcommand{\eeqn}{ \end{eqnarray}}

\makeatletter 
\g@addto@macro{\@algocf@init}{\SetKwInOut{parameter}{Tuning parameter}} 
\makeatother

\usepackage{color}

\theoremstyle{plain} 
\newtheorem{thm}{Theorem}

\newtheorem{prop}{Proposition}
\newtheorem{lem}{Lemma}

\theoremstyle{remark}

\title{Valid model-free spatial prediction}
\author{Huiying Mao\footnote{The Statistical and Applied Mathematical Sciences Institute} \quad 
   Ryan Martin\footnote{North Carolina State University} \quad Brian J Reich$^\dagger$}
\date{\today}

\begin{document}


\maketitle

\begin{abstract}
Predicting the response at an unobserved location is a fundamental problem in spatial statistics.  Given the difficulty in modeling spatial dependence, especially in non-stationary cases, model-based prediction intervals are at risk of misspecification bias that can negatively affect their validity. Here we present a new approach for model-free nonparametric spatial prediction based on the {\em conformal prediction} machinery.  Our key observation is that spatial data can be treated as exactly or approximately exchangeable in a wide range of settings.  In particular, under an infill asymptotic regime, we prove that the response values are, in a certain sense, locally approximately exchangeable for a broad class of spatial processes, and we develop a local spatial conformal prediction algorithm that yields valid prediction intervals without strong model assumptions like stationarity. Numerical examples with both real and simulated data confirm that the proposed conformal prediction intervals are valid and generally more efficient than existing model-based procedures for large datasets across a range of non-stationary and non-Gaussian settings.

\smallskip


\end{abstract}

\section{Introduction}\label{s:intro}


Providing valid predictions of the response at an unobserved location is a fundamental problem in spatial statistics.  For example, epidemiologists may wish to extrapolate air pollution concentrations from a network of stationary monitors to the residential locations of the study participants.  There are a number of challenges one faces in carrying out valid prediction at a new spatial location, but one of the most pressing is that existing methods are model-based, so the reliability of the predictions depends crucially on the soundness of the posited model. For example, prediction intervals based on Kriging---see \citet{cressie1992statistics} and Section~\ref{s:back:spatial}---often rely on normality, and stationarity is often assumed to facilitate estimating the spatial covariance function required for Kriging. It is now common to perform geostatistical analysis for massive datasets collected over a vast and diverse spatial domain \citep[e.g.,][]{heaton2019case}.  For complex processes observed over a large domain, the normality and stationarity assumptions can be questionable. Failing to account for nonstationarity can affect prediction accuracy, but typically has a larger effect on uncertainty quantification such as prediction intervals \citep{fuglstad2015does}.  While there are now many methods available for dealing with nonstationary \citep[see][for a recent review]{risser2016nonstationary} and non-Gaussianity  \citep[e.g.,][]{gelfand2005bayesian,duan2007generalized,reich2007multivariate,rodriguez2011nonparametric}, these typically involve heavy computations.  This exacerbates the already imposing computational challenges posed by massive datasets. Further, fitting the entire stochastic process may be unnecessary if only prediction intervals are desired.  Nonparametric machine-learning methods can be used for prediction \citep{kim2016accurate, kim2016deeply, lim2017enhanced, tai2017image,hengl2018random,franchi2018supervised,wang2019nearest,li2020deepkriging}, but these methods typically focus on uncertainty estimation. In this paper we propose a method with provably valid prediction intervals---exact in some cases, asymptotically approximate in others---for the response at a single location without requiring specification of a statistical model, and hence not inheriting the risk of model misspecification bias.



In recent years, the use of machine learning techniques in statistics has become increasingly more common.  While there are numerous examples of this phenomenon, the one most relevant here is {\em conformal prediction}. This method originated in \citet{vovk2005algorithmic} and the references therein \citep[see, also,][]{shafer2008tutorial}, but has appeared frequently in the recent statistics literature \citep[e.g.,][]{lei2014distribution, lei2018distribution, guan2019conformal, romano2019conformalized, conformal.shift}. What makes this method especially attractive is that it provides provably valid prediction intervals without specification of a statistical model.  More precisely, the conformal prediction intervals achieve the nominal frequentist prediction coverage probability, uniformly over all data distributions; see  Section~\ref{s:back_conformal}.  The crucial assumption behind the validity of conformal prediction is that the data are exchangeable.  

Whether it is reasonable to assume exchangeability in a spatial application depends on how the data are sampled.  On the one hand, if the locations are randomly sampled in the spatial domain, then exchangeability holds automatically; see Lemma~\ref{lem:iid.locations}.  In such cases, standard conformal prediction can be used basically off the shelf.  On the other hand, if the locations are fixed in the spatial domain, then exchangeability does not hold in general.  We show, however, that for a wide range of spatial processes, the response variables at tightly concentrated locations are approximately exchangeable; see Theorem~\ref{thm:local.exchangeable}.  Therefore, a version of the basic conformal prediction method applied to these tightly concentrated observations ought to be approximately valid.


Using this insight about the connection between exchangeability and the sampling design, we propose two related spatial conformal prediction methods.  The first, a so-called {\em global spatial conformal prediction} (GSCP) method, described in Section~\ref{s:GSCP}, is designed specifically for cases where the spatial locations are sampled at random.  In particular, this global method produces a prediction interval which is marginally valid, i.e., valid on average with respect to the distribution of the target location at which prediction is desired; asymptotic efficiency of this global method is also investigated.  The second, a {\em local spatial conformal prediction} (LSCP) method, described in Section~\ref{s:LSCP}, is designed specifically for the case when the spatial locations are fixed.  Since our goal is to proceed without strong assumptions about the spatial dependence structure, it is only possible to establish approximate or local exchangeability.  Therefore, the proposed local spatial conformal prediction method can only provide approximately valid predictions; see Theorem~\ref{thm:valid.local}.  But our goal in the fixed-location case resembles the ``conditional validity'' target in the conformal prediction literature \citep[e.g.,][]{barber2019limits, chernozukov.dconformal} so, given the impossibility theorems in the latter context, approximate validity is all that can be expected.  

For both the global and local formulations, our proposed method is computationally feasible for large datasets and model-free in the sense that its validity does not depend on a correctly-specified model. In Sections~\ref{s:sim} and \ref{s:app}, we show using real and simulated data that the proposed methods outperform both standard global Kriging and local approximate Gaussian process regression \citep[laGP;][]{gramacy2015local} for non-stationary and non-Gaussian data.  In addition to be useful for spatial applications, it is also an advancement in conformal prediction to the case of dependent data, and establishes the conditions on the spatial sampling design and data-generating mechanism that ensure (approximate) validity of the conformal prediction intervals.      

The remainder of this paper is organized as follows. Section~\ref{s:back} reviews spatial and conformal prediction. Sections \ref{s:GSCP} and \ref{s:LSCP} introduce the proposed methods, which are examined using simulations in Section \ref{s:sim} and a real data analysis in Section \ref{s:app}. Additional numerical and theoretical results, along with all proofs, are given in the Supplemental Materials. 
 
\section{Background} \label{s:back}

\subsection{Spatial prediction methods} \label{s:back:spatial}

Let $Y_i\in \RR$ and $X_i \in \RR^d$, with $d \geq 1$, be the observable pairs at spatial location $s_i\in\calD\subseteq \RR^2$.  Note that $X_i$ can include covariates that are deterministic functions of the spatial location, such as elevation, genuinely stochastic covariates like wind speed, or even non-spatial covariates such as the smoking status of the resident at location.  Write the data points as triples $Z_i = (s_i, X_i, Y_i)$, for $i=1,2,\ldots$.  We assume only a single observation is made at each location and thus often adopt the notation $Y_i = Y(s_i)$ and $X_i = X(s_i)$.

Geostatistical analysis often assumes that the data follow a Gaussian process model, $Y_i = X_i^\top \beta + \theta_i + \varepsilon_i$, for $ i=1,\ldots,n$, where $\beta$ is the vector of regression coefficients, $\varepsilon_1,\ldots,\varepsilon_n$ are independent $\mbox{Normal}(0,\tau^2)$ errors, $\theta_i = \theta(s_i)$, and $\theta$ is a mean-zero Gaussian process with isotropic covariance function $C(\theta_i,\theta_j) = \sigma^2 \rho(d_{ij})$, a function of the distance $d_{ij}$ between locations $s_i$ and $s_j$.  A common example is the $\Matern$ correlation function $\rho(d; \phi, \kappa)$ parameterized  by  correlation range $\phi$ and smoothness parameter $\kappa$. Denote the spatial covariance parameters as $\Theta=\{\sigma^2,\tau^2,\phi, \kappa\}$. The main assumptions of this model are that the data are Gaussian and the covariance function is stationary and isotropic, i.e., it is a function only of the distance between spatial locations and is thus the same across the spatial domain.

Consider data $Z^{n+1} = (Z_1,\ldots,Z_n, Z_{n+1})$.  In our applications, $Z^n$ will be observed and $(s_{n+1}, X_{n+1})$ will be given, and the goal is to predict the corresponding $Y_{n+1}$.  However, the ordering of the data is irrelevant so one can imagine different orderings that correspond to data point $i$ in the last position, where $i=1,\ldots,n,n+1$.  That is, imagine we have the observed data $z_{(i)}^{n+1} = \{z_1,\ldots,z_{n+1}\} \setminus \{z_i\}$, along with $(s_i, x_i)$ and parameter estimates $\widehat\beta$ and $\widehat\Theta$; then the predictive distribution of $Y_i$ is normal with mean 
$\hat\mu_{n+1,i}(s_i, x_i)$ and variance $\sigma_{n+1,i}^2(s_i, x_i)$, where both the mean and variance depend on ${\widehat \Theta}$ and the configuration of the spatial locations; see the Supplementary Material for the specific expressions.  The standardized residuals are 
\begin{equation}
\label{eq:residual}
e_{n+1,i} = \frac{y_{i} - \hat\mu_{n+1,i}(s_i, x_i)}{\hat\sigma_{n+1,i}(s_i, x_i)}, \quad i=1,\ldots,n,n+1,
\end{equation}
and the corresponding $100(1-\alpha)\%$ prediction interval for $Y_i$ is
$ \hat\mu_{n+1,i}(s_i, x_i) \pm q_\alpha^\star \, \hat\sigma_{n+1,i}(s_i, x_i)$,
where $q_\alpha^\star$ is the upper $\alpha/2$ quantile of a standard normal distribution.

\subsection{Conformal prediction} \label{s:back_conformal}

Here we take a step back and review conformal prediction for non-spatial problems; for a detailed treatment, see \cite{vovk2005algorithmic} and \cite{shafer2008tutorial}. Suppose we have a data sequence $Z_1,\ldots, Z_n, Z_{n+1},\ldots$, assumed to be exchangeable with joint distribution $\prob$, that is, $Z_1,Z_2,Z_3,\dots$ and $Z_{\xi(1)}, Z_{\xi(2)},Z_{\xi(3)},\ldots$ have the same joint distribution for any permutation $\xi$ defined on the positive integers. This data may be response-only, i.e., $Z_i = Y_i$, or may be response-covariate pairs, i.e., $Z_i = (X_i, Y_i)$; we will focus on the latter more general case.  No assumptions about $\prob$ are made here, beyond that it is exchangeable. We observe $Z^n=z^n$, and the goal is to predict $Y_{n+1}$ at a new value $X_{n+1}$ of the covariate.  More specifically, we seek a procedure that returns, for any $\alpha \in (0,1)$, a prediction interval  $\Gamma^\alpha(Z^n; X_{n+1})$ that is {\em valid} in the sense that  
\begin{equation}
\label{eq:valid}
\prob^{n+1}\{\Gamma^\alpha(Z^n; X_{n+1}) \ni Y_{n+1} \} \geq 1-\alpha, \quad \text{for all $(\alpha, n, \prob)$}, 
\end{equation}
where $\prob^{n+1}$ is the distribution of $(Z_1,\ldots,Z_n,Z_{n+1})$ under $\prob$.  That we require the inequality \eqref{eq:valid} to hold for {\em all exchangeable distributions $\prob$} rules out the use of model-based procedures, such as likelihood or Bayesian methods.   

The original conformal prediction method proceeds as follows.  Define a {\em non-conformity measure} $\Delta(B,z)$, a function that takes two arguments: the first is a ``bag'' $B$ that consists of a finite collection of data points; the second is a single data point $z$.  Then $\Delta(B,z)$ measures how closely $z$ represents the data points in bag $B$.  For example, if $\pi$ is a prediction rule and $d$ is some measure of distance, then we might take $\Delta(B,z) = d(\pi_B, z)$, the distance between $z$ and the value $\pi_B$ returned by the prediction rule $\pi$ applied to $B$.  The choice of $\Delta$ depends on the context, though often there is a natural choice.  Throughout this paper we will assume $\Delta$ is symmetric in its first argument, so that shuffling the data in bag $B$ does not change the value of $\Delta(B,z)$.  

Given the non-conformity measure $\Delta$, the next step is to appropriately transform the data via $\Delta$.  Specifically, augment the observed data $z^n = (z_1, \ldots, z_n)$ with a provisional value $z_{n+1}$ of $Z_{n+1} = (X_{n+1}, Y_{n+1})$; this $z_{n+1}$ value is generic and free to vary.  Define 
\[ \delta_i = \Delta(z_{(i)}^{n+1}, z_i), \quad i=1,\ldots,n,n+1. \]
Note that $\delta_{n+1}$ is special because it compares the actual observed data with this provisional value of the unobserved future observation.  Next, compute the {\em plausibility} \citep[e.g.,][]{cella2020valid} of $z_{n+1}$ as a value for $Z_{n+1}$ according to the formula 
\begin{equation}
\label{eq:plausibility}
p(y_{n+1} \mid z^n, x_{n+1}) = \frac{1}{n+1} \sum_{i=1}^{n+1} 1\{\delta_i \geq \delta_{n+1}\},
\end{equation}
where $1\{A\}$ denotes the indicator of event $A$.  Note that this process can be carried out for any provisional value $z_{n+1}$, so the result is actually a mapping $\tilde y \mapsto p(\tilde y \mid z^n, \tilde x)$, for a given $\tilde x$, which we will refer to as the {\em plausibility contour} returned by the conformal algorithm.  This function can be plotted 
to visualize the uncertainty about $Y_{n+1}$ based on the given $\tilde x$, data $z^n$, the choice of non-conformity measure, etc.  Moreover, a prediction set $\Gamma^\alpha(z^n; \tilde x)$ can be obtained as 
\begin{equation}
\label{eq:preset.review}
\Gamma^\alpha(z^n; \tilde x) = \{\tilde y: p(\tilde y \mid z^n, \tilde x) > \alpha\}. 
\end{equation}

\section{Global spatial conformal prediction} \label{s:GSCP}

\subsection{GSCP algorithm}

The first approach we consider is a direct application of the original conformal algorithm to spatial prediction but with spatial dependence encoded in the non-conformity measure.  In contrast to the local algorithm presented in Section~\ref{s:LSCP}, this method equally weights the non-conformity across all spatial locations in the plausibility contour evaluation.  Therefore, we refer to this as global spatial conformal prediction, or GSCP for short. 

From Section~\ref{s:back:spatial}, recall that $z_i=(s_i, x_i, y_i)$ and $z^{n+1} = \{z_1,\ldots,z_{n+1}\}$; also, $z_{(i)}^{n+1}$ denotes $z^{n+1} \setminus \{z_i\}$, the full data set with $z_i$ excluded. Now define the non-conformity measure for the GSCP algorithm as
\begin{equation}
\label{e:A_stat}
 \delta_i = \Delta(z_{(i)}^{n+1}, z_i) = \Bigl| \frac{y_i - \hat\mu_{n+1,i}(s_i, x_i)}{\hat\sigma_{n+1,i}(s_i, x_i)} \Bigr|,
 \quad i=1,\ldots,n+1, 
\end{equation}
where $\hat\mu_{n+1,i}(s_i,x_i)$ and $\hat\sigma_{n+1,i}(s_i, x_i)$ are, respectively, the mean response and standard error estimates at spatial location $s_i$ with covariate $x_i$ based on data in $z_{(i)}^{n+1}$. Then we define a plausibility contour exactly like in Section~\ref{s:back_conformal}, with obvious notational changes.  That is, for provisional values $(s_{n+1}, x_{n+1}, y_{n+1})$ of $(S_{n+1}, X_{n+1}, Y_{n+1})$, we have  
\begin{equation}\label{eq:GSCP:plausibility}
 p(y_{n+1} \mid z^n; s_{n+1}, x_{n+1}) = \frac{1}{n+1} \sum_{i=1}^n 1\{\delta_i \geq \delta_{n+1}\}. 
\end{equation}
The corresponding $100(1-\alpha)$\% prediction interval for $Y_{n+1}$, denoted by $\Gamma^\alpha(z^n; s_{n+1}, x_{n+1})$, is just an upper level set of the plausibility contour, consisting of all those provisional $y_{n+1}$ values with plausibility exceeding $\alpha$, analogous to \eqref{eq:preset.review}. 

Any reasonable choice of $(\hat\mu,\hat\sigma)$ estimates can serve the purpose here, including inverse distance weighting predictions \citep{henley2012nonparametric}, deep learning predictions \citep{franchi2018supervised}, 
Kriging predictions, etc. 
In our numerical results presented below, we use the Kriging estimates as defined in Section~\ref{s:back:spatial}, so that $\delta_i$ is a standardized Kriging residual, $|e_i|$, from \eqref{eq:residual}. Conformal prediction is invariant to monotone transformations of its $\delta_i$'s, and we found that similar results are obtained with other related measures, such as unstandardized Kriging residuals.  A particular advantage of our recommended choice of $\delta_i$'s is that we can quickly compute the plausibility contour and prediction interval by exploiting the inherent quadratic structure of the Kriging-based non-conformity measure; see the Supplementary Materials.
Moreover, note that validity of the GSCP-based prediction intervals does not require the Gaussian model associated with the Kriging method be correctly specified, nor does it depend on our choice of the $\delta_i$'s.

\subsection{Theoretical validity of GSCP}
\label{s:theory_global}

Given the importance of exchangeability to the validity of conformal prediction and the fact that the spatial dependence generally is incompatible with exchangeability, we might have some concerns about the validity of GSCP.  However, there are practically relevant cases in which exchangeability does hold, in particular, when the spatial locations are sampled independently and identically distributed (iid).  The following elementary lemma explains this.

\begin{lem}
\label{lem:iid.locations}
If the spatial locations $S_1,S_2,\ldots$ are iid, then $Z_1,Z_2,\ldots$, with $Z_i = (S_i, X(S_i), Y(S_i))$, is an exchangeable sequence.  
\end{lem}

Since randomly sampled spatial locations makes the data exchangeable, a validity property for GSCP follows immediately from the general theory in, e.g., \citet{shafer2008tutorial}.  

\begin{thm}
\label{thm:valid.global}
Let $(X,Y)$ be a stochastic process over $\calD$ and let $S_1,S_2,\ldots$ be iid draws in $\calD$.  Let $Z_i = (S_i, X(S_i), Y(S_i))$ for $i=1,2,\ldots$, and define the coverage probability function 
\[ c(\alpha, n, \prob) = \prob^{n+1}\{\Gamma^\alpha(Z^n; X_{n+1}, S_{n+1}) \ni Y(S_{n+1})\}. \]
Then the proposed GSCP is valid in the sense that
\begin{equation}
\label{e:mv}
c(\alpha, n, \prob) \geq 1-\alpha, \quad \text{for all $(\alpha,n,\prob)$}, 
\end{equation}
where $\prob^{n+1}$ is the joint distribution of $Z_1,\ldots,Z_n,Z_{n+1}$ under $\prob$.  Moreover, if $\delta_1,\ldots,\delta_{n+1}$ in \eqref{e:A_stat} have a continuous distribution, then 
\begin{equation}
\label{e:mv.new}
c(\alpha,n,\prob) \leq 1-\alpha + (n+1)^{-1}, \quad \text{for all $(\alpha,n,\prob)$}.
\end{equation}
\end{thm}

The upper bound in \eqref{e:mv.new}, which follows from the same arguments as in \citet{lei2018distribution}, implies that the GSCP method is not only valid but also efficient in the sense that the coverage probability is not too much larger than the nominal level.  That is, the coverage condition is not being achieved simply giving excessively wide intervals. Some further details on the efficiency of the global spatial conformal prediction procedure are investigated in  the Supplementary Materials.  
Note, also, that Theorem~\ref{thm:valid.global} makes no assumptions about the distribution of $(X,Y)$, so it surely covers non-Gaussian and non-stationary processes.  

Theorem~\ref{thm:valid.global} gives a marginal validity result in the sense that it accurately predicts the response $Y(S_{n+1})$ at $X(S_{n+1})$, for a randomly sampled spatial location $S_{n+1}$.  However, it does not ensure conditional validity, i.e., the case where $S_{n+1}=s^\star$ with $s^\star$ being a fixed spatial location.  There are negative results in the literature \citep[e.g.,][]{lei2014distribution} which state that strong conditional validity---for all $\prob$ and almost all targets $s^\star$---is impossible with conformal prediction.  Considerable effort has been expended recently trying to achieve ``approximate'' conditional validity in some sense; see, e.g., \citet{lei2014distribution}, \citet{conformal.shift}, \citet{barber2019limits}, and \citet{chernozukov.dconformal}.  Remarkably, there is at least one scenario in which a strong conditional validity result can be achieved in our context.  In particular, as we show in the Appendix, the GSCP-based intervals are both marginally and conditionally valid for the special case of an isotropic process sampled uniformly on a sphere.  
Admittedly, these  are rather strong conditions, so one would hope for (approximately) valid prediction under much less.  In Section~\ref{s:LSCP}, we show that asymptotically valid prediction intervals at a fixed location can be obtained under only mild conditions on the sampling scheme and the unknown process.

\section{Local spatial conformal prediction} 
 \label{s:LSCP}


\subsection{LSCP algorithm}
 \label{SS:local.algorithm}

For valid prediction at a fixed location $s^\star$, we propose a local spatial conformal prediction (LSCP) approach that is based on only those data points in the neighborhood of $s^\star$. Fix an integer $m > 0$ and select a neighborhood around $s^\star$ that contains $m$ many locations, $s_{i_j}$, for $j = 1,\ldots,m$.  Note that $\{s_{i_j}: j=1,\ldots,m\}$ is a subset of the full set of spatial locations $s_1,\ldots,s_n$.  Without structural assumptions about the response process, such as stationarity, the data at locations far from $s^\star$ are not obviously relevant to prediction at $s^\star$, so removing---or down-weighting (Section~\ref{SS:smoothed})---them from the local analysis is reasonable.  Plus, in applications where the infill asymptotic regime is appropriate, there are many observations nearby $s^\star$, so $m$ could be taken to be large. 

From here, we can proceed very much like in Section~\ref{s:GSCP}.  For notational simplicity, assume that indices $i=1,\ldots,m$ correspond to those $m$ spatial locations closest to $s^\star$.  Now let $Z_i = (s_i, X_i, Y_i)$, for $i=1,\ldots,m$, denote the observations at these $m$ closest locations to $s^\star$.  With a slight abuse of that notation, set $s_{m+1} = s^\star$ and $(x_{m+1}, y_{m+1})$ as the provisional values of $X$ and $Y$ at $s^\star$.  Then define the non-conformity scores exactly as before: 
 \[ \delta_i = \Delta(z_{(i)}^{m+1}, z_i), \quad i=1,\ldots,m+1. \]
 With this, we can readily obtain the plausibility contour function:
 \begin{equation} 
 \label{eq:LSCP:plausibility}
 p(y_{m+1} \mid z^m, s^\star, x_{m+1}) = \frac{1}{m+1} \sum_{i=1}^{m+1} 1\{\delta_i \geq \delta_{m+1}\}. 
 \end{equation} 
 Specific details are presented in Algorithm~2 in the Supplementary Materials. 
 The output of this algorithm is a $100(1-\alpha)$\% prediction interval for $Y_{m+1} = Y(s^\star)$, depending on $Z^m$ and the observed $X_{m+1}=X(s^\star)$, which we denote by $\Gamma_{s^\star}^\alpha(Z^m; X_{m+1})$.

 \subsection{Theoretical validity of LSCP}\label{SS:LSCP_theory}

 Our theoretical results hinge on a definition of local exchangeability.  Let $\calD \subset \RR^2$ be a compact spatial domain, e.g., $[0,1]^2$.  For a generic $\RR^d$-valued stochastic process $T$ defined on $\calD$, with $d \geq 1$, define the {\em localized version} of $T$, relative to a location $s^\star \in \calD$, as 
 \begin{equation}
 \label{eq:local}
 \widetilde T_r(u) = T(s^\star + r u), \quad u \in \calU = \{u \in \RR^2: \|u\| \leq 1\},
 \end{equation}
 indexed by the unit disk $\calU$ and the radius $r > 0$. Now suppose that $T$ can be decomposed as 
 \begin{equation}
 \label{eq:decomp}
 T(s) = \psi\bigl( L(s), E(s) \bigr), \quad s \in \calD, 
 \end{equation}
 where $L$ and $E$ are independent $\RR^d$-valued stochastic process, $L$ is a continuous spatial process, $E$ is a non-spatial process, and $\psi$ is a deterministic, continuous, $\RR^d$-valued function.  More specifically, suppose that $L$ and $E$, respectively, satisfy the following conditions:
 \begin{itemize}
 \item $L$ is {\em $L_2$-continuous} at $s^\star$ in the sense that its localized version $\widetilde L_r$ satisfies $\E\|\widetilde L_r(u) - \widetilde L_0(u)\|^2 \to 0$ as $r \to 0$ for any $u \in \calU$;  
 \item $E$ is {\em locally iid} at $s^\star$, that is, its localized version $\widetilde E_r$ converges in distribution to an iid process as $r \to 0$.    
 \end{itemize}
This formulation is too abstract to be useful, but formulating a general result here is appropriate.  Appendix~A.2 describes several common spatial models that satisfy \eqref{eq:decomp}, and generalizes the above formulation to the case where the covariates are also considered stochastic processes.

These assumptions yield a certain kind of {\em local exchangeability} which will be used below to show the LSCP algorithm achieves a desired validity property.  We first establish this local exchangeability result, which may be of independent interest.  

 \begin{prop}
 \label{thm:local.exchangeable}
 Suppose that $T$ can be decomposed as in \eqref{eq:decomp}, where $L$ is $L_2$-continuous at $s^\star$, $E$ is locally iid at $s^\star$, and $L$ and $E$ are independent. Then the localized process $\widetilde T_r$ in \eqref{eq:local} converges in distribution as $r \to 0$, and the limit is an exchangeable process in the sense that its finite-dimensional distributions are exchangeable.
\end{prop}

Using Proposition~\ref{thm:local.exchangeable}, we can establish the (asymptotically approximate) theoretical validity of the LSCP method.  To set the scene, those $m$ spatial locations closest to $s^\star$ fall in a neighborhood of some radius $r$. 
As is common in the spatial statistics literature \citep[e.g.,][]{stein1990comparison,cressie1992statistics}, we adopt an {\em infill asymptotic regime} in which the region $\calD$ remains fixed while the number of observations $n$ goes to infinity, hence filling the space.  The relevant point for our analysis is that under this regime the number of observations made in any neighborhood of $s^\star \in \calD$ will go to infinity.  Such a regime is natural---and necessary---in cases without structural assumptions about $Y$, e.g., stationarity, where it is simply not possible to learn the local features of a process at $s^\star$ if data are not concentrated in a neighborhood around $s^\star$.  Under the infill asymptotic framework , if $m$ is fixed and the number of locations $n$ is increasing to fill the bounded space $\calD$, then the radius of the neighborhood in which those $m$ points fall is vanishing.  For example, if the spatial locations are (roughly) uniformly distributed in $\calD$, then the number of points in a neighborhood of radius $r$ would be  proportional to $n r^2$; setting this equal to $m$ gives $r=r_n = (m/n)^{1/2} \to 0$ as $n \to \infty$.

 It follows from Proposition~\ref{thm:local.exchangeable} that the joint distribution of the response $Y$ at these $m$-many spatial locations around $s^\star$ (corresponding to $m$-many vectors in $\calU$) would be approximately exchangeable and, consequently, a conformal prediction algorithm that creates non-conformity scores using only these $m$ observations would be valid for predicting $Y(s^\star)$.  

 \begin{thm}
 \label{thm:valid.local}
 Consider an infill asymptotic regime with $n$ spatial locations in the bounded domain $\calD$, with $n \to \infty$.  Fix an integer $m > 0$ and let $r=r_n \to 0$ be such that the $m$ closest locations to $s^\star$ fall in a neighborhood of radius $r$.  Under the assumptions of Proposition~\ref{thm:local.exchangeable}, the non-conformity measure $\Delta$ is a continuous function of its inputs, and if the limiting distribution in Proposition~\ref{thm:local.exchangeable} is continuous, then the LSCP prediction intervals are asymptotically valid at $s^\star$ in the sense that 
 \[ \lim_{n \to \infty} \prob^{m+1}\{ \Gamma_{s^\star}^\alpha(Z^m; X_{m+1}) \ni Y_{m+1} \} = 1-\alpha + O(m^{-1}), \]
 where $\prob^{m+1}$ is the joint distribution of $Z_1,\ldots,Z_m,Z_{m+1}$ at the $m$ spatial locations and at $s^\star$.
 \end{thm}

 \subsection{The smoothed LSCP algorithm}
 \label{SS:smoothed}

 Theorem~\ref{thm:valid.local} implies that the local spatial conformal prediction with $m$ nearest neighbors is approximately valid under the infill asymptotic regime. However, in practice completely disregarding the contribution of the observations outside the $m$ nearest neighbors may be unsatisfactory, so we propose a smoothed version of the LSCP algorithm (sLSCP). 

The GSCP algorithm weights all $n+1$ non-conformity measures $\delta_i$ equally in the plausibility contour computation in \eqref{eq:plausibility}, but this is questionable for non-stationary processes with  stochastic properties that vary throughout the spatial domain.  To allow for non-stationarity, the sLSCP algorithm weights the non-conformity measures $\delta_i$ by how far the corresponding $s_i$ is from the prediction location. Let $f$ be a non-increasing function, and define weights 
 \[ w_i \propto f(d_i), \quad i=1,\ldots,n+1, \]
 where $d_i = \|s_i-s^\star\|$, $d_{n+1} \equiv 0$, and the proportionality constant ensures that $\sum_{i=1}^{n+1} w_i = 1$.  Different $f$ functions can be applied, but we recommend the normalized Gaussian kernel function with bandwidth $\eta$, 
\begin{equation}
 \label{eq:weight}
 w_{i} = \frac{\exp(-d_{i}^2/2\eta^2)}{1+\sum_{j=1}^{n}\exp(- d_{j}^2/2\eta^2)}, \quad i=1,\ldots,n+1. 
 \end{equation}
Note that, if $\eta \to \infty$, then $w_i \to (n+1)^{-1}$ for each $i$, which corresponds to the GSCP algorithm.  Finally, with these new weights, the plausibility contour at a provisional value $(s_{n+1}, x_{n+1}, y_{n+1})$ of $Y(s^\star)$ is given by
 \begin{equation}\label{eq:sLSCP:plausibility}
   p_w(y_{n+1} \mid Z_{(n+1)}^{n+1}, s_{n+1}, x_{n+1}) = \sum_{i=1}^{n+1} w_i 1\{\delta_i \geq \delta_{n+1}\},
 \end{equation}
As before, we recommend the Kriging-based strategy with $\delta_i$ the  standardized residual in \eqref{eq:residual}.

Since we are interested only in the local structure of $Y$, it is natural that locations far from $s^\star$ have negligible weight, as in \eqref{eq:weight}.  But including all $n$ observations requires some non-trivial calculations, e.g., inverting a large $n \times n$ covariance matrix. Therefore, to avoid cumbersome and ultimately irrelevant computation, we recommend using only the $M \ll n$ observations closest to $s^\star$ for both the Kriging predictions that determine $\delta_i$ and in the plausibility scores in \eqref{eq:sLSCP:plausibility}.  The resulting method is both locally adaptive and computationally efficient even for large data sets. 

The tuning parameter $\eta$ can be selected using cross validation, as illustrated in Sections \ref{s:sim} and \ref{s:app}.  The value of $M$ is determined by the bandwidth $\eta$ so that all observations with substantial $w_i$ are included, as are observations that are required for the Kriging prediction of these observations. Typically the number of nearby observations required to approximate the Kriging prediction is a small subset of the total number of observations \citep{stein2002screening}.  As a rule of thumb, $M$ could be selected to roughly include all observations within $2\eta + r^\star$ radius of $s^\star$, where $2\eta$ captures observations with substantial weights, and $r^\star$ is selected so that all the $M$ observations include the nearest 15 neighbors of the observation within $2\eta$ of $s^\star$. We summarize the details of Algorithm sLSCP in Algorithm~\ref{ag:lscp+}.  For simplicity, we use sLSCP and LSCP indistinguishably.

 \begin{algorithm}[t]
 \SetAlgoLined
 \KwIn{observations $z_i = (s_i, x_i, y_i), i=1, \ldots, n$; predict location $s^\star$; non-conformity measure $\Delta$; significance level $\alpha$; and a fine grid of candidate response values}
 \parameter{weight parameter $\eta \in (0, \infty)$; number of neighbors to consider $M \leq n$}
 \KwOut{$(1-\alpha)100\%$ prediction interval, $\Gamma^\alpha$, for $Y(s^\star)$}
 determine $M$ through $\eta$ if not given\;
 form $z_i, i = 1, \ldots, M,$ based on $M$ locations closest to $s^\star$\;
 $s_{M+1} \gets s^\star$\;
 calculate weights $w_i, i=1, \ldots, M+1$ as in \eqref{eq:weight}\;
 \For{$y_{M+1}$ in the specified grid}{
 \For{i = 1 to M+1}{
 define $z_{(i)}^{M+1}$ by removing $y_i$ from $z^{M+1}$\;
 $\delta_i \gets \Delta(z_{(i)}^{M+1}, y_i)$ \;}
 compute plausibility for $y_{M+1}$ as $p_w(y_{M+1} \mid \cdots)$ in \eqref{eq:sLSCP:plausibility}\;
 include $y_{M+1}$ in $\Gamma^\alpha$ if $p_w(y_{M+1} \mid \cdots) \geq t_M(\alpha)$\;}
 \Return $\Gamma^\alpha$.
 \caption{Smoothed local spatial conformal prediction (sLSCP).}\label{ag:lscp+}
 \end{algorithm}

It is important for our proposed methods to be computationally feasible for large datasets. Conformal prediction itself is relatively expensive since it requires fitting the underlying model once for each held-out data point being predicted.  In particular, the Kriging residuals and the associated non-conformity score computations require us to compute $\hat\mu_{n+1,i}(s_i, x_i)$ and $\hat\sigma_{n+1,i}(s_i, x_i)$ for each $i=1,\ldots,n$, which involves $n$ many evaluations of $(\widehat\beta, \widehat\Theta)$. To overcome this computational bottleneck, various adjustments have been considered in the literature.  One is the {\em split conformal prediction} strategy---also called {\em inductive conformal prediction} in \citet{vovk2005algorithmic}---which is common; see, e.g., \citet[][Sec.~2.2]{lei2018distribution}. The idea is to split the data into two parts: one for fitting the underlying model and the other for running conformal prediction with the fitted model from the first part fixed. The theoretical validity of split conformal prediction is now well-known, e.g., Section~3 of the Supplementary Materials. Alternatively, as is common in parametric Kriging, one could use the entire dataset to estimate $(\widehat\beta, \widehat\Theta)$ and then use the entire dataset again for prediction with the parameter estimates plugged in as if they were the ``true values.''  Given that the number of parameters in the working spatial model is relatively small, both approaches should perform well for moderate to large $n$. The simulation results presented in the Supplemental Materials suggest that this plug-in conformal is more efficient than split conformal in terms of width of the corresponding prediction intervals (or, more precisely, in terms of the interval score as defined in Section~\ref{s:sim:methods}), so the numerical results in Sections~\ref{s:sim} and \ref{s:app} below are based on plug-in versions of the proposed GSCP and LSCP algorithms.

\section{Simulation study}\label{s:sim}


\subsection{Data generation} \label{s:sim:data}

We consider one mean-zero Gaussian stationary process (Scenario~1) and seven non-Gaussian and/or nonstationary data-generating scenarios (Scenarios~2--8). Data are generated based on transformations of a latent Gaussian process $Z(s)$ and a white noise process $E(s)$ with standard normal distribution, where $s=(s_x,s_y) \in [0,1]^2$. The mean-zero stationary Gaussian process process $Z(s)$ has a $\Matern$ covariance function with variance $\sigma^2=3$, range $\phi=0.1$, and smoothness $\kappa = 0.7$. Data are sampled on the $N\times N$ grid of $n=N^2$ points in the unit square, $s\in \{N^{-1}, 2 N^{-1}, \ldots, 1\}^2$, with $N = 20$ or $N = 40$.  The scenarios are:
\begin{enumerate}[label=\arabic*, labelsep=0pt]
    \item \label{scen:stationary}. $Y(s) = Z(s) + E(s)$;
    \item \label{scen:cubic}. $Y(s) = Z(s)^3 + E(s)$;
    \item \label{scen:skewed}. $Y(s) = q[\Phi\{Z(s)/\sqrt{3}\}] + E(s)$ where $\Phi$ is the standard normal distribution function and $q$ is the $\text{Gamma}(1,3^{-1/2})$ quantile function;    
    \item \label{scen:prod}. $Y(s) = \sqrt{3} Z(s)  |E(s)|$;
    \item \label{scen:gfun}. $Y(s) = \mbox{sign}\{Z(s)\}  |Z(s)|^{s_x+1}+E(s)$;
    \item \label{scen:eastwest}. $Y(s) = \sqrt{\omega(s) / 3}\, Z(s) + \sqrt{1-\omega(s)} \, E(s)$ where $\omega(s)=\Phi(\frac{s_x-0.5}{0.1})$;
    \item \label{scen:nugget}. $Y(s) = Z(s) + s_x \, E(s)$; 
    \item \label{scen:spike}. $Y(s) = Z(s) + 10 \exp(-50\|s-c\|^2)$ where $c = (0.5,0.5)$; 
\end{enumerate}
Scenario~\ref{scen:stationary} is Gaussian and stationary, Scenarios~\ref{scen:cubic}--\ref{scen:prod} are stationary but non-Gaussian, and Scenarios~\ref{scen:gfun}--\ref{scen:spike} are  nonstationary either in the spatial variance (Scenarios~\ref{scen:gfun} and \ref{scen:eastwest}), error variance (Scenario~\ref{scen:nugget}), or mean (Scenario~\ref{scen:spike}). Scenario \ref{scen:skewed} generates skewed data to assess the method's performance when the symmetry of the base Kriging model is violated.  



 

\subsection{Prediction methods and metrics}\label{s:sim:methods}

For each dataset we apply the global and local (with $\eta=0.1$) conformal spatial prediction algorithms. For the parametric Kriging method and the initial Kriging predictions of our proposed conformal prediction, we estimate the spatial covariance parameters using empirical variogram methods \citep{cressie1992statistics}. The empirical variograms are calculated using the {\tt variog} function in the {\tt R} package {\tt geoR}, and the covariance parameters are chosen to minimize the weighted (by number of observations) squared error between the empirical and model-based variograms.  

We compare the proposed conformal prediction methods with standard global Kriging prediction and the local Kriging (laGP) method of \cite{gramacy2015local} that dynamically defines the support of a Gaussian process predictor based on a local subset of the data. For laGP, we use the function provided by the {\tt laGP} package in {\tt R} the local sequential design scheme starting from 6 points to 50 points through an empirical Bayes mean-square prediction error criterion.  

Methods are trained using a completely random set of 90\% of the observations and tested on the remaining 10\%.  Each scenario is repeated 100 times, and performance is evaluated using average coverage of $(1-\alpha)100$\% prediction intervals, average interval width, and average interval score \citep{gneiting2007strictly}, defined as  
$$S_\alpha(I;y^n) = \frac{1}{n}\sum_{i=1}^n\bigl\{ (I_u-I_l) + \tfrac{2}{\alpha} (I_l-y_i)_+ + \tfrac{2}{\alpha}(y_i-I_u)_+  \bigr\},$$ where $I=[I_l,I_u]$ is the $100(1-\alpha)\%$ prediction interval, $y^n$ contains the observations $y_1,\ldots,y_n$, and $z_+ = z \vee 0$ denotes the ``positive part.''  A smaller interval score is desirable as this rewards both high coverage and narrow intervals. We use $\alpha=0.1$ in this simulation study.  

\subsection{Results}
\label{s:sim:results}

We present results averaged over data sets and all spatial locations in Table~\ref{t:sim_stationary}. For the non-stationary scenarios varying across $s_x$ (Scenarios \ref{scen:gfun}--\ref{scen:nugget}), we present the results by the first spatial coordinate ($s_x$) averaged over the data sets and the second coordinate ($s_y$) in Figure~\ref{f:gfun_n_eastwest}, e.g., the value of coverage plotted at $s_x=N^{-1}$ is the average of the coverage over the $N$ points of the form $(N^{-1},s_y)$ for $s_y\in\{N^{-1}, 2 N^{-1}, \ldots, 1\}$.

\begin{table}
\caption{\label{t:sim_stationary}Performance comparison for simulation scenarios (``Scen'') without a covariate. The metrics are the empirical coverage of 90\% prediction intervals (``Cov90''), the width of prediction intervals (``Width'') and the interval score (``IntScore''), each averaged over location and dataset. The methods are global (GSCP) and local (LSCP) conformal prediction, stationary and Gaussian Kriging (``Kriging'') and local approximate Gaussian process (``laGP'') regression. }
\begin{center}
\begin{tabular}{ll|ccc|ccc}
 & & \multicolumn{3}{c|}{$N=20$} & \multicolumn{3}{c}{$N=40$}   \\
 Scen & Method & Cov90 & Width & IntScore & Cov90 & Width & IntScore \\ 
  \hline
\ref{scen:stationary}
& GSCP & 0.906 & 4.67 & 5.78    & 0.897 & 4.00 & 5.06 \\               
& LSCP & 0.890 & 4.57 & 5.95    & 0.891 & 3.99 & 5.12 \\               
& Kriging & 0.912 & 4.73 & 5.78 & 0.888 & 3.90 & 5.07 \\            
& LaGP & 0.877 & 4.78 & 6.50    & 0.879 & 4.27 & 5.63 \\     
&&&&&&&\vspace{-6pt}\\
\ref{scen:cubic}
& GSCP & 0.895 & 33.62 & 67.16 & 0.897 & 22.12 & 43.80 \\           
& LSCP & 0.896 & 31.05 & 58.37 & 0.910 & 21.74 & 36.47 \\           
& Kriging & 0.931 & 44.07 & 69.38 & 0.924 & 27.75 & 44.80 \\        
& LaGP & 0.913 & 40.57 & 69.38 & 0.928 & 31.28 & 47.24 \\           
&&&&&&&\vspace{-6pt}\\
\ref{scen:skewed}
& GSCP & 0.908 & 4.79 & 6.32 & 0.895 & 4.05 & 5.27 \\               
& LSCP & 0.893 & 4.65 & 6.27 & 0.893 & 4.03 & 5.24 \\               
& Kriging & 0.919 & 4.95 & 6.33 & 0.887 & 3.96 & 5.28 \\            
& LaGP & 0.883 & 4.92 & 6.83 & 0.880 & 4.29 & 5.77 \\               
&&&&&&&\vspace{-6pt}\\
\ref{scen:prod} 
& GSCP & 0.902 & 7.06 & 11.06 & 0.895 & 6.25 & 10.25 \\             
& LSCP & 0.892 & 6.84 & 11.23 & 0.895 & 6.08 &  9.74 \\             
& Kriging & 0.918 & 7.70 & 11.12 & 0.908 & 6.71 & 10.26 \\          
& LaGP & 0.898 & 7.40 & 11.83 & 0.901 & 6.68 & 10.51 \\             
&&&&&&&\vspace{-6pt}\\
\ref{scen:gfun} 
& GSCP & 0.900 & 6.51 &  9.40 & 0.898 & 5.03 & 7.11 \\              
& LSCP & 0.887 & 6.36 &  9.06 & 0.892 & 5.05 & 6.75 \\              
& Kriging & 0.924 & 7.18 &  9.52 & 0.897 & 5.11 & 7.15 \\           
& LaGP & 0.887 & 7.20 & 10.42 & 0.891 & 5.98 & 8.08 \\              
&&&&&&&\vspace{-6pt}\\
\ref{scen:eastwest} 
& GSCP & 0.894 & 2.78 & 3.69 & 0.896 & 2.63 & 3.60 \\               
& LSCP & 0.878 & 2.66 & 3.47 & 0.895 & 2.38 & 3.06 \\               
& Kriging & 0.897 & 2.78 & 3.69 & 0.888 & 2.54 & 3.61 \\            
& LaGP & 0.865 & 2.58 & 3.60 & 0.869 & 2.32 & 3.21 \\               
&&&&&&&\vspace{-6pt}\\
\ref{scen:nugget} 
& GSCP & 0.905 & 3.63 & 4.55 & 0.896 & 2.77 & 3.71 \\              
& LSCP & 0.888 & 3.53 & 4.59 & 0.896 & 2.70 & 3.46 \\               
& Kriging & 0.915 & 3.77 & 4.55 & 0.889 & 2.70 & 3.72 \\            
& LaGP & 0.869 & 3.92 & 5.31 & 0.881 & 3.17 & 4.14 \\               
&&&&&&&\vspace{-6pt}\\
\ref{scen:spike} 
& GSCP & 0.906 & 3.04 & 3.74 & 0.899 & 1.93 & 2.45 \\               
& LSCP & 0.880 & 3.00 & 3.91 & 0.895 & 1.92 & 2.46 \\               
& Kriging & 0.928 & 3.25 & 3.78 & 0.915 & 2.05 & 2.48 \\            
& LaGP & 0.863 & 3.39 & 4.64 & 0.871 & 2.41 & 3.20 \\  
\end{tabular}
\end{center}
\end{table}

In Scenario~\ref{scen:stationary}, the Gaussian and stationary process, the performance of GSCP, LSCP, and Kriging are comparable (Table~\ref{t:sim_stationary}). Kriging performs well in this case since the data generating mechanism aligns with its underlying assumption, but the conformal methods are competitive with the parametric model in terms of both coverage and interval width. In Scenarios~\ref{scen:cubic}, \ref{scen:skewed}, and \ref{scen:prod}, the non-Gaussian but stationary processes, GSCP, LSCP, and Kriging perform more or less the same in terms of interval score and outperform laGP. However, the coverage of the conformal methods, especially the GSCP algorithm, is closer to the nominal level than Kriging and the Kriging intervals are generally wider than the conformal intervals.

Figure \ref{f:gfun_n_eastwest} shows the results for nonstationary Scenarios \ref{scen:gfun} and \ref{scen:eastwest} when $N=40$. 
LSCP performs the best among the four methods for these nonstationary scenarios. For Scenario~\ref{scen:gfun}, the process is Gaussian to the west and non-Gaussian to the east. The global prediction methods GSCP and Kriging generate prediction intervals with similar width for all $s_x$ (ignoring edge effects), while LSCP and laGP provide wider intervals on the east ($s_x$ near 1) than the west ($s_x$ near 0). LSCP has coverage around 90\% for all $s_x$, and the lowest interval scores, especially in the east where the process is more non-Gaussian.  Similarly, in Scenario~\ref{scen:eastwest}, the correlation is stronger in the east than the west, and the LSCP performs the best by providing adaptive prediction interval width and valid coverage across space. 

\begin{figure}
\centering 
\includegraphics[width=0.4\textwidth,page=1]{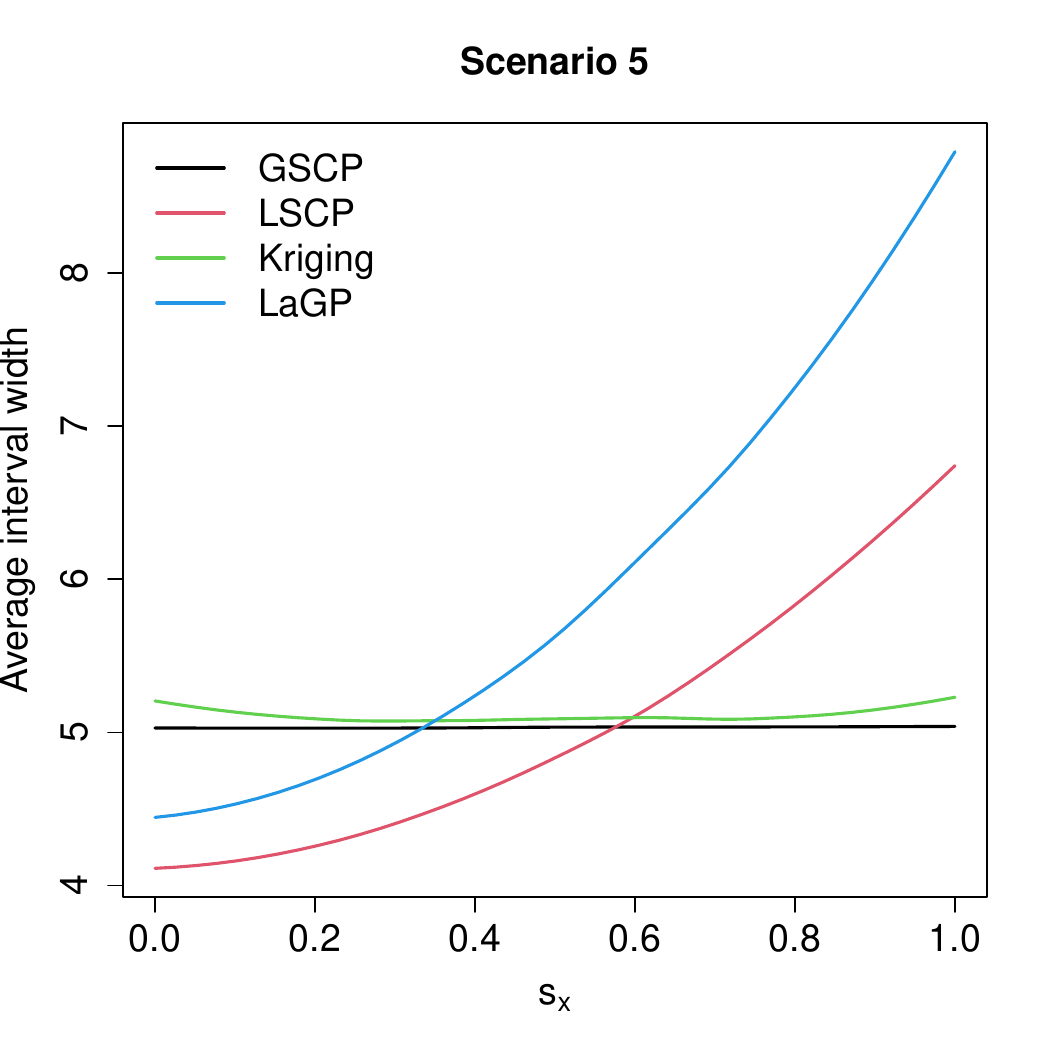}
\includegraphics[width=0.4\textwidth,page=4]{figs/east_west_revision.pdf}
\includegraphics[width=0.4\textwidth,page=2]{figs/east_west_revision.pdf}
\includegraphics[width=0.4\textwidth,page=5]{figs/east_west_revision.pdf}
\includegraphics[width=0.4\textwidth,page=3]{figs/east_west_revision.pdf}
\includegraphics[width=0.4\textwidth,page=6]{figs/east_west_revision.pdf}
\caption{Performance comparison by $s_x$ for Scenario \ref{scen:gfun}: $Y(s) = \mbox{sign}\{Z(s)\} \cdot |Z(s)|^{s_x+1}+E(s)$ and Scenario \ref{scen:eastwest}: $Y(s) = \sqrt{\omega(s)/3}\cdot Z(s) + \sqrt{1-\omega(s)}\cdot E(s)$ where $\omega(s)=\Phi(\frac{s_x-0.5}{0.1})$ when $N=40$ (results are smoothed over $s_x$ for clarity). }
\label{f:gfun_n_eastwest}
\end{figure}

We also conducted a simulation study when spatial locations are sampled uniformly on $[0,1]^2$. The performance is very similar to that when locations are fixed at equally-spaced grid points, so we only show the latter in the paper. Additional results for the scenarios with covariates (thus a comparison with universal Kriging) and a sensitivity analysis confirming our method's robustness to the estimates of the spatial covariance parameters are included in the Supplemental Materials.


\section{Real data analysis}\label{s:app}

This section demonstrates the performance of conformal prediction method using the canopy height data in Figure~\ref{f:canopy}a. The data were originally presented in \cite{cook2013nasa} and were analyzed using a nearest-neighbor Gaussian process model in \cite{datta2016hierarchical}. The data are available in the R package {\tt spNNGP} \citep{finley2017spnngp}. There are $n=1,723,137$ observations and clear nonstationarity and non-normality.  For example, there are several heterogeneous areas with small height canopies around the location with longitude and latitude being 729,000 and 470,000, respectively. 


We compare methods using 90\% prediction intervals for 10,000 test locations chosen randomly from the full dataset.  Since the data clearly exhibit nonstationarity we do not apply GSCP.  We select the kernel function and bandwidth parameter using cross-validation over the validation locations. The average interval score is consistently smaller for the Gaussian kernel (ranging between 4.3 and 4.4 by $\eta$) than the uniform kernel (ranging between 4.9 and 5.2 by $\eta$) and minimized by the Gaussian kernel with $\eta=6\times10^{-4}$.
Table~\ref{tab:canopy_perf_comp} compares the performance of LSCP on the 10,000 test locations with Kriging and laGP.  LSCP outperforms the other methods as the empirical coverage of LSCP is the closet to the desired 90\% and the LSCP minimizes the interval score.

\begin{table}
\caption{\label{tab:canopy_perf_comp}Performance comparison for the canopy height data. The metrics are the width, coverage (``Cov90'') and interval score (``IntScore'') of 90\% prediction intervals, each averaged over 10,000 randomly chosen test locations. The methods are local  conformal prediction (LSCP), stationary and Gaussian Kriging (``Kriging'') and local approximate Gaussian process (``laGP'') regression.}
    \centering
    \begin{tabular}{l|ccc}
             & Width & Cov90   & IntScore\\
     \hline
     LSCP    & 2.87  & 87.9\%  & 4.33 \\
     Kriging & 5.44  & 96.6\%  & 6.81 \\
     laGP    & 5.04  & 91.1\%  & 6.63 \\
    \end{tabular}
\end{table}

Figure \ref{f:canopy} plots the interval widths for each method. Unlike Kriging, the LSCP and laGP interval widths are locally adaptive with wider intervals in heterogeneous areas and more narrow intervals in homogeneous areas. Comparing LSCP and laGP, LSCP generally provides narrower intervals than laGP, which means the proposed method is more efficient than laGP. In addition, the locations of the observations that fall outside the prediction intervals are uniformly distributed for LSCP and laGP, but clustered in the heterogeneous areas for Kriging.  In short, the proposed spatial conformal prediction algorithm shows its superiority in this real data analysis. 

\begin{figure}
(a) Original data
\hspace{100pt}(b) LSCP result\\
\includegraphics[width=0.49\linewidth]{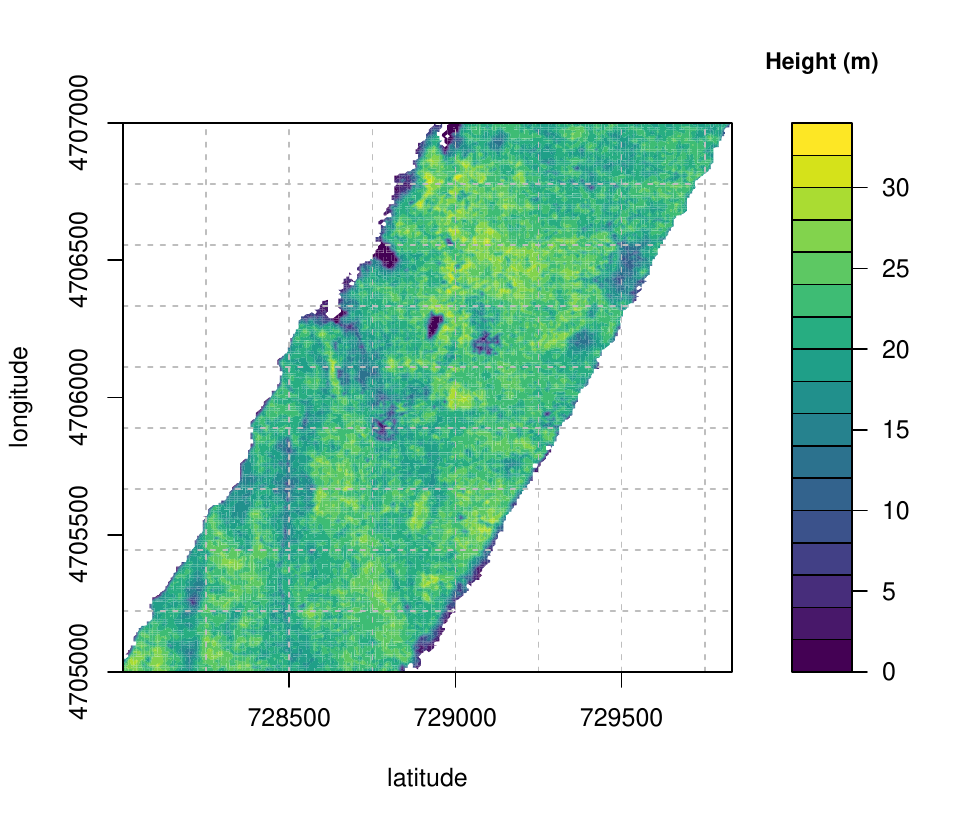}
\includegraphics[width=0.49\linewidth]{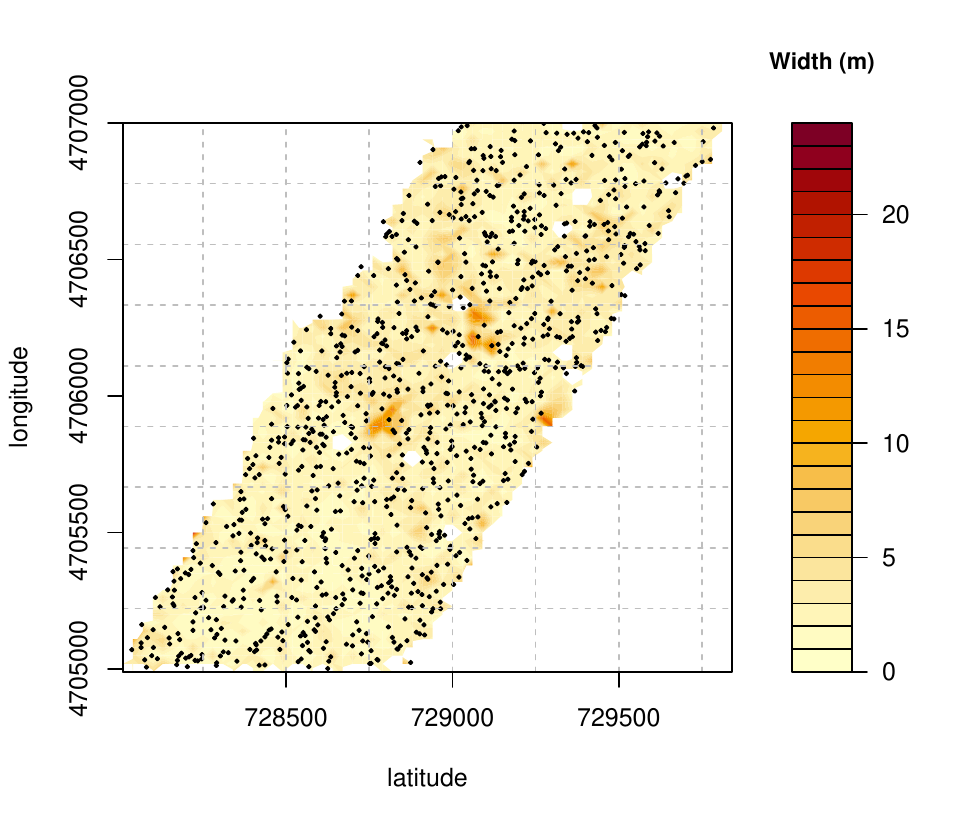}\\
(c) Kriging result
\hspace{100pt}(d) laGP result\\
\includegraphics[width=0.49\linewidth]{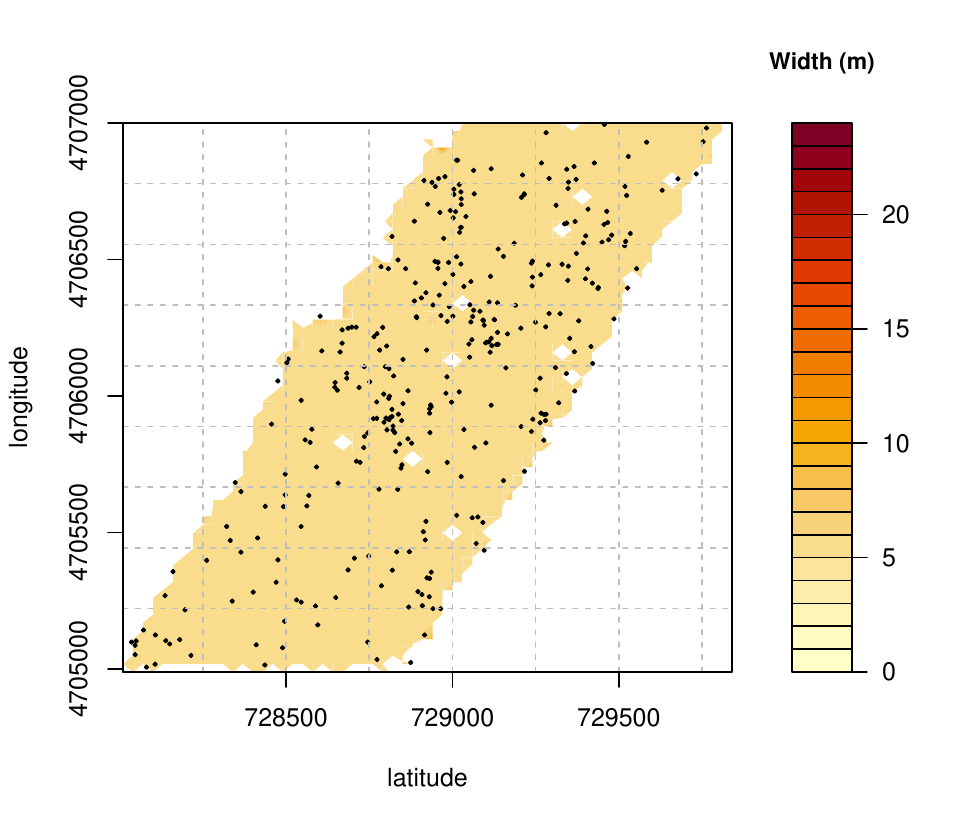} 
\includegraphics[width=0.49\linewidth]{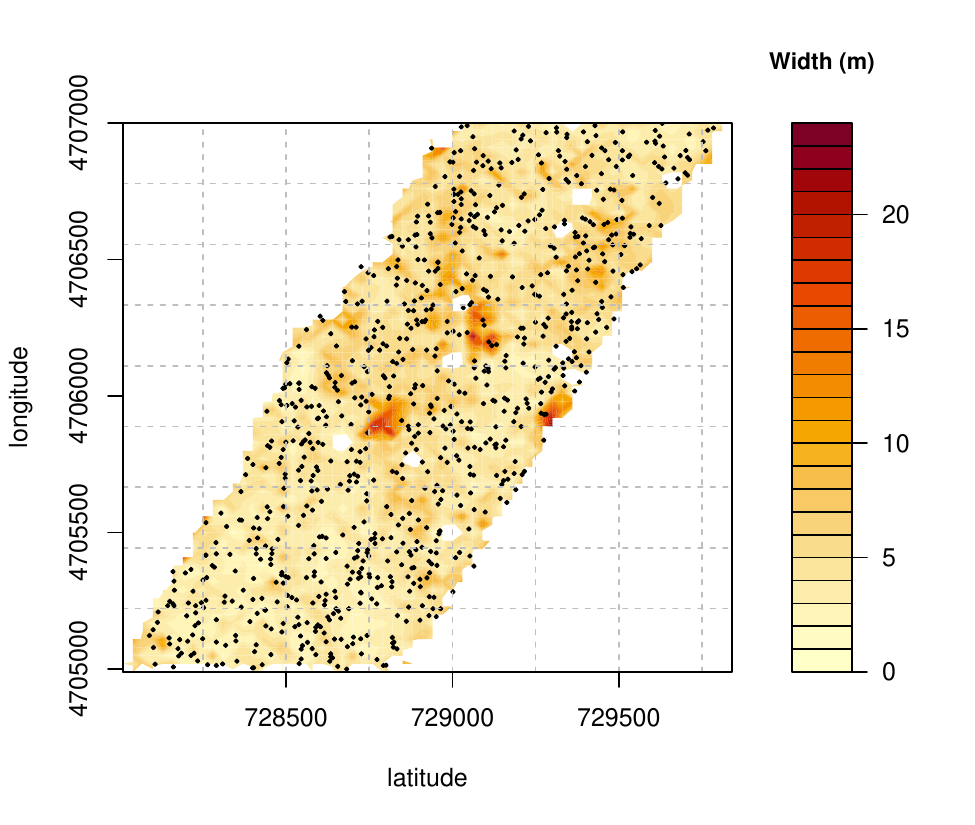} 
\caption{(a) Heatmap of the canopy height data; (b)--(d) Prediction interval width (color) and locations not covered (points) for LSCP, Kriging, and laGP. Longitude and latitude are in UTM Zone 18. }
\label{f:canopy}
\end{figure}

\section{Discussion}\label{s:discussion}

In this paper, we proposed a spatial conformal prediction algorithm to provide valid, robust, and model-free prediction intervals by combining spatial methods and the classical conformal prediction. We provided both global and local versions to accommodate different stationarity cases and sampling designs. We proved their validity under various sampling designs and data-generating mechanisms. To the authors' knowledge, this work is among the first in making the classical conformal algorithm work for non-exchangeable data. Our simulation studies and real data analyses demonstrate the advantage of the proposed spatial conformal prediction algorithms. We also developed an {\tt R} package entitled {\tt scp} (\url{https://github.com/mhuiying/scp}) to compute the plausibility contours and generate spatial prediction intervals using either Kriging residual or any other user-defined non-conformity measure.

An attractive feature of the proposed algorithms is that they are model-free in the sense that their theoretical validity does not depend on correct specification of a model.  In our implementations we use the squared residuals from a simple parametric model to define the non-conformity terms. However, as anticipated by our theoretical results, our simulation study shows that the methods work well even if the parametric family or the mean and covariance functions are misspecified, and that the results are insensitive to inaccurate estimation of the parameters in the parametric model. This robustness allows the methods to be applied broadly and with confidence.

The local conformal prediction method relies on a dense grid of points around each prediction location and thus a large dataset.  Computation time often prohibits application of spatial methods to large datasets.  Fortunately, we are able to apply our method to large datasets by exploiting local algorithms and an explicit formula for the plausibility contour using the Kriging-based non-conformity score.  Similar derivations are needed for prediction procedures other than Kriging in order to maintain computational efficiency.

Future research directions include extending the work to spatial processes with discrete observations, e.g., when $Y(s)$ is binary or a count.  Generalized spatial linear models are compatible with our current framework (see the Appendix), but continuity in the distribution function is required.  Therefore, further studies would be required to establish the validity of conformal prediction for discrete data. Another limitation of the proposed algorithms is that they only produce intervals for a single location.  Generalizing the algorithms to produce joint intervals for multiple locations would be useful in some applications.  One option is to use a Bonferroni correction; of course, this may be inefficient for many simultaneous predictions, but greater efficiency would require model assumptions to link the multiple locations and facilitate information sharing. It would also be of interest to extend the proposed spatial conformal prediction methods to spatiotemporal data, perhaps building on recent work for time series data \citep{xu2020conformal,
zaffran2022adaptive}.


\section*{Appendix}

\subsection*{A.1: Conditional validity on a sphere}

An obstacle that prevents a conditional validity result in the existing literature is an ``edge effect.''  That is, conditional validity is typically achieved at targets in the middle of the domain, but fails at targets in the extremes; see Figure~1(b) in \citet{lei2014distribution}.  So if it were possible to eliminate the edge effect---even if in a trivial way, by eliminating the edge itself---then there is hope for establishing a conditional validity result.  In our spatial context, but perhaps not in other cases, it may not be unreasonable to assume that the spatial locations are sampled iid from a uniform distribution on a sphere.  Since the sphere has no edges and a uniform distribution has no extremes, there is no ``edge effect'' preventing conditional validity.  Some additional structure in the $(X,Y)$ process is also needed here, in particular, it should be isotropic in the sense that the correlation structure only depends on the distance between spatial locations.  Note that, if the mean of the parametric base model is correctly specified, so that the conditional distribution of $Y-X\beta$, given $X$, is free of $X$, then the stationarity assumption about $X$ can be removed. 

To our knowledge, Proposition~\ref{prop:sphere} below gives the first finite-sample conditional validity result for conformal prediction in the literature, albeit under rather strong conditions.  

\begin{prop}
\label{prop:sphere}
Let $(X,Y)$ be an isotropic stationary process over the sphere $\calD = \{s \in \RR^3: \|s\|=1\}$, and suppose that the locations $S_1,\ldots,S_n, S_{n+1}$ are independent and uniformly distributed on $\calD$. For $\Gamma^\alpha$ as described above, define the conditional coverage probability function 
\[ c(s^\star \mid \alpha, n, \prob) = \prob^{n+1}\{\Gamma^\alpha(Z^n; S_{n+1}, X_{n+1}) \ni Y_{n+1} \mid S_{n+1} = s^\star\}. \]
Then the GSCP-based predictions are conditionally valid, i.e., $c(s^\star \mid \alpha, n, \prob) \geq 1-\alpha$ for all $(\alpha, n, s^\star)$ and all $\prob$ under which $(X,Y)$ is stationary and isotropic and $S$ are iid uniform on $\calD$. 
\end{prop}


\subsection*{A.2: Locally-exchangeable processes}

To better understand how the locally-exchangeable processes in Section~\ref{SS:LSCP_theory} relate to our prediction problem, consider a simple case with no covariates, where $\{Y(s): s \in \calD\}$ is the only random process under consideration.  In that case, we want to show that $T(s) = Y(s)$ has this local exchangeability property.  Then the sufficient condition \eqref{eq:decomp} above amounts to assuming there exists a suitable real-valued function $\psi_Y$, along with appropriate processes $L_Y$ and $E_Y$, such that 
\[ Y(s) = \psi_Y\bigl( L_Y(s), E_Y(s) \bigr), \quad s \in \calD. \]
There are a number of common models for continuous responses that meet this condition, including the additive model in Section~\ref{s:back:spatial}, certain generalized spatial linear models \cite{diggle1998model}, spatial copula models \citep{krupskii2019copula}, and max-stable processes \citep{reich2012hierarchical}.  For example, in a generalized spatial linear model, with suitable spatial process $L_Y$ and Gaussian white noise $E_Y$, take
\begin{equation} 
\label{eq:glm}
\psi_Y(\ell,e) = H_{g(\ell)}^{-1}\bigl( \Phi^{-1}(e) \bigr),  
\end{equation} 
where $H_\xi$ is the distribution function for an exponential family with natural parameter $\xi$, $g$ is the link function, and $\Phi$ is the standard normal distribution function.  Of course, to meet the continuity requirement, the exponential family must have a density with respect to Lebesgue measure.  


For the practically relevant case with both a response $Y$ and covariate $X$ process, the idea is similar but the notation is more complicated.  The goal is to find conditions under which the joint process $T(s) = (X(s),Y(s))$ has a representation as in \eqref{eq:decomp}.  Admittedly, it is challenging to consider the joint process directly, which is why we aim to give a simpler sufficient condition based on the marginal distribution of $X$ and the conditional distribution of $Y$, given $X$.  Consider the following decomposition:
\begin{equation}
\label{eq:covariate.decomp}
\begin{split}
X(s) & = \psi_X\bigl( L_X(s) , E_X(s) \bigr) \\
Y(s) & = \psi_Y\bigl( L_Y(s) , E_Y(s) \mid X(s) \bigr).
\end{split}
\end{equation}
Roughly, this amounts to assuming that each of $X$ and $Y$ has a decomposition like that described above for $Y$ alone.  Individual assessments of the distributional properties of $X$ and $Y$ are more manageable than directly considering their joint distribution.  And as the following simple lemma states, separate considerations of its marginal and conditional structure suffice to establish a decomposition of the joint structure.

\begin{lem}
\label{lem:decomp}
If $(X,Y)$ can be decomposed as in \eqref{eq:covariate.decomp}, if the pairs of processes $(L_X,L_Y)$ and $(E_X,E_Y)$ in are individually $L_2$-continuous and locally iid, respectively, and if $(\ell, e) \mapsto \psi_X(\ell, e)$ and $(\ell, e, x) \mapsto \psi_Y(\ell, e \mid x)$ are both continuous, then $T(s) = (X(s), Y(s))$ satisfies the conditions of Proposition~\ref{thm:local.exchangeable}. 
\end{lem}


As we discussed above, the decomposition of the $X$ marginal as in \eqref{eq:covariate.decomp} is quite flexible.  For example, it is quite common that $X$ could be expressed as $X=L_X + E_X$ for a spatial and non-spatial components, $L_X$ and $E_X$, respectively, but the additive form is not necessary.  And just like in our discussion above Proposition~\ref{prop:sphere} above, if the mean of the parametric base model is correctly specified, then these assumptions about $X$ here can be dropped. Similarly, if the decomposition of $Y$ in the response-only model was flexible, then the corresponding conditional decomposition in \eqref{eq:covariate.decomp} must be equally flexible.  For example, the same generalized spatial linear model can be considered, but now it is allowed to depend smoothly on the covariate.

\section*{Supplementary materials}

\subsection*{Computational details}\label{s:A1}

We describe the computational algorithm for the case of no covariates.  If there are covariates, the algorithm below is applied to the residuals and the covariate effects are added back to the resulting prediction interval.  We first give prediction equation for the base Kriging model. To make predictive inference at the new location $s_{n+1}$, we study the random variable $Y_{n+1} = Y(s_{n+1})$ together with the other $Y_i$'s.  The full data set 
is denoted by $Z^{n+1} = \{Z_1,...,Z_n, Z_{n+1}\}$.  The collection of data excluding $Z_i$ is denoted by $Z_{(i)}^{n+1}$.  The joint distribution of $Y^{n+1}=(Y_1,...,Y_n, Y_{n+1})$, given $(s_1,...,s_n,s_{n+1})$, the $(n+1)\times d$ covariate matrix $X^{n+1}$ and the spatial covariance parameters $\Theta=\{\sigma^2,\tau^2,\phi, \kappa\}$, is 
$ Y^{n+1} \sim \mbox{Normal}\big(X^{n+1}\beta,\Sigma(\Theta)\big)$,
where $\Sigma(\Theta)$ is a $(n+1)\times(n+1)$ covariance matrix with $(i,j)$ element 
$ \sigma^2\rho(d_{ij}; \phi, \kappa) + I(i=j)\tau^2$.
Let ${\widehat \beta}$ and ${\widehat \Theta}$ be estimates of $\beta$ and $\Theta$, respectively, and ${\widehat Q} = \Sigma({\widehat \Theta})^{-1} = \{{\hat q}_{ij}\}$.  The Kriging prediction and variance of $Y_i$, given $\beta={\widehat \beta}$, $\Theta  = {\widehat \Theta}$ and $Z_{(i)}^{n+1}$, are
\begin{equation}
\label{eq:kriging_supp}
\begin{split}
\hat\mu_{n+1,i}(s_i, x_i) & = \E(Y_i \mid Z_{(i)}^{n+1}, s_i, x_i) 
 = x_i^\top{\widehat \beta} -\hat q_{ii}^{-1} \sum_{j\ne i}{\hat q}_{ij}(y_j- x_j^\top{\widehat \beta})\\
\hat \sigma_{n+1,i}^2(s_i, x_i) & = \var(Y_i \mid Z_{(i)}^{n+1}, s_i, x_i) = {\hat q}_{ii}^{-1}, 
\end{split}
\end{equation}
for each $i=1,\ldots,n,n+1$.  

Applying the predictive distributions in \eqref{eq:kriging_supp} and provisionally setting $Y_{n+1}=y$, for $i\ne n+1$ we can write
\begin{eqnarray}\label{e:quad}
\delta_i-\delta_{n+1} &=& q_{ii}\left(Y_i-{\tilde Y}_{i,n+1}+\frac{q_{i,n+1}}{q_{ii}}y\right)^2 - q_{n+1,n+1}\left(y-{\hat Y}_{n+1}\right)^2 \\
&=& U_{i} + V_{i}y+W_{i}y^2\nonumber
\end{eqnarray}
where ${\tilde Y}_{i,n+1} = -\sum_{j\ne\{i,n+1\}}q_{ij}Y_j/q_{ii}$, $U_{i} = q_{ii}(Y_i-{\tilde Y}_{i,n+1})^2-q_{n+1,n+1}{\hat Y}_n^2$, $V_{i} = 2q_{i,n+1}(Y_i-{\tilde Y}_{i,n+1}) + 2q_{n+1,n+1}{\hat Y}_{n+1}$ and $W_{i} = q_{i,n+1}^2/q_{ii} -q_{n+1,n+1}$. 

Since the inverse covariance matrix $Q$ is positive definite, $q_{ii}>0$ and $q_{ii}>\max_{j\ne i}|q_{ij}|$, so $W_{i} <0$ and $V_i^2-4U_iW_i = \frac{4q_{n+1,n+1}}{q_{ii}}[\sum_{i=1}^n q_{ij}Y_j+q_{i,n+1}\hat Y_{n+1}]^2\geq0$. Therefore, the $y$ satisfying $\delta_i-\delta_{n+1}\geq0$ are within two roots, denoted as $a_i \leq b_i$, of the quadratic equation $U_{i} + V_{i}y+W_{i}y^2=0$. Then the plausibility calculation \eqref{eq:GSCP:plausibility} becomes \begin{eqnarray*}
p(y|Z_{(n+1)}^{n+1})&=&\frac{1}{n+1}\sum_{i=1}^{n}1\{U_{i} + V_{i}y+W_{i}y^2\geq0\}+\frac{1}{n+1}\\ &=& \frac{1}{n+1}\sum_{i=1}^{n}1\{a_i \leq y \leq b_i\}+\frac{1}{n+1},\end{eqnarray*} 
which is a step function with jumping points being $a_i$'s and $b_i$'s, and the plausibility function for sLSCP in \eqref{eq:sLSCP:plausibility} simply becomes 
$\sum_{i=1}^{n}w_i1\{a_i \leq y \leq b_i\}$. With the help of the step function, we can solve for the prediction interval directly with no need to enumerate for possible $y$ satisfying $p(y|Z_{(n+1)}^{n+1})\geq t_n(\alpha)$. The GSCP and LSCP steps are summarized in Algorithms \ref{ag:gscp} and \ref{ag:lscp}, and smoothed LSPC algorithm is given in the main text.

\begin{algorithm}[t]
\SetAlgoLined
\KwIn{observations $z_i = (s_i, x_i, y_i)$, $i=1,\ldots,n$; prediction location and covariate $(s_{n+1}, x_{n+1})$; non-conformity measure $\Delta$; significance level $\alpha$; and a fine grid of candidate $y_{n+1}$ values}
\KwOut{$100(1-\alpha)\%$ prediction interval, $\Gamma^\alpha = \Gamma^\alpha(z^n; s_{n+1}, x_{n+1})$, for the response $Y_{n+1} = Y(s_{n+1})$ at $x_{n+1} = X(s_{n+1})$.}

\For{provisional values $y_{n+1}$ in the specified grid}{
\For{$i=1,\ldots,n+1$}{ 
$z_{(i)}^{n+1} \gets z^{n+1} \setminus \{z_i\}$\; 
$\delta_i \gets \Delta(z_{(i)}^{n+1}, z_i)$\; 
}
compute plausibility for $y_{n+1}$ as $p(y_{n+1} \mid z^n, s_{n+1}, x_{n+1})$ in \eqref{eq:GSCP:plausibility}\;
include $y_{n+1}$ in $\Gamma^\alpha$ if $p(y_{n+1} \mid z^n, s_{n+1}, x_{n+1}) \geq t_n(\alpha)$\;
}
\Return $\Gamma^\alpha$.

\caption{Global spatial conformal prediction (GSCP).} \label{ag:gscp}
\end{algorithm}

\begin{algorithm}[t]
\SetAlgoLined
\KwIn{observations $z_i = (s_i, x_i, y_i), i=1,\ldots,n$; prediction location and covariate $(s^\star, x_{m+1})$; non-conformity measure $\Delta$; significance level $\alpha$; and a fine grid of candidate $y^\star$ values}
\parameter{number of neighbors to consider $m \leq n$}
\KwOut{$100(1-\alpha)\%$ prediction interval, $\Gamma_{s^\star}^\alpha = \Gamma_{s^\star}^\alpha(z^m; x_{m+1})$, for $Y(s^\star)$ with $X(s^\star) = x_{m+1}$: }

form $z_i, i = 1, \ldots, m$, based on $m$ locations closest to $s^\star$\; 
$s_{m+1} \gets s^\star$\;
\For{provisional values $y_{m+1}$ in the specified grid}{
\For{$i = 1$ to $m+1$}{
$z^{m+1}_{(i)} \gets z^{m+1} \setminus \{z_i\}$\;
$\delta_i \gets \Delta(z^{m+1}_{(i)}, z_i)$\;}
compute plausibility $p(y_{m+1} \mid z^m, s^\star, x_{m+1})$ in \eqref{eq:LSCP:plausibility}\;
include $y_{m+1}$ in $\Gamma_{s^\star}^\alpha$ if $p(y_{m+1} \mid z^m, s^\star, x_{m+1}) \geq t_m(\alpha)$\;}
\Return $\Gamma_{s^\star}^\alpha$.

\caption{Local conformal spatial prediction (LSCP).}\label{ag:lscp}
\end{algorithm}

\subsection*{Proofs from Sections~\ref{s:GSCP}--\ref{s:LSCP}}

\subsubsection*{Proof of Proposition 1}
\label{proof:local.exchangeable}

Write the localized version of the $T$ process as 
\[ \widetilde T_r(u) = \psi \bigl( \widetilde L_r(u), \widetilde E_r(u) \bigr), \quad u \in \calU, \]
where $\psi$ is a continuous function of two arguments and the $L$ and $E$ components are $L_2$-continuous and locally iid, respectively.  We are interested in the finite-dimensional distribution of the localized process, so take a fixed set of $m$ vectors $u_1,\ldots,u_m$ in $\calU$.  Define the vectors 
\[ \widetilde L_r^m = \bigl( \widetilde L_r(u_1),\ldots, \widetilde L_r(u_m) \bigr) \quad \text{and} \quad \widetilde E_r^m = \bigl( \widetilde E_r(u_1), \ldots, \widetilde E_r(u_m) \bigr). \]
First, a bit of notation.  If $w$ is a generic object, let $w^{\otimes m}$ denote a copy of $m$ versions of $w$ in in rows.  For example, if $w$ is a scalar, then $w^{\otimes m} = w 1_m$, where $1_m$ is an $m$-vector of unity.  Alternatively, if $w$ is a (column) vector, then $w^{\otimes m}$ is a matrix with $m$ identical rows, each containing $w^\top$.  Also, let $\|\cdot\|$ denote either the usual Euclidean norm on $\RR^m$ or the Frobenius norm on matrices with $m$ rows, depending on the dimension of its argument.  

For the $L$ part in the above representation, Markov's inequality implies
\[ \prob\{\|\widetilde L_r^m - L(s^\star)^{\otimes m}\| > \varepsilon\} \leq \varepsilon^{-2} \, \E\|\widetilde L_r^M - L(s^\star)^{\otimes m}\|^2, \quad \text{for any $\varepsilon > 0$}. \]
The expectation in the upper bound can be rewritten as 
\begin{align*}
\E\|\widetilde L_r^m - L(s^\star)^{\otimes m}\|^2 & = \sum_{i=1}^m \E\|\widetilde L_r(u_i) - L(s^\star)\|^2 \\
& = \sum_{i=1}^m \E\|L(s^\star + r u_i) - L(s^\star)\|^2, 
\end{align*}
and, since $m$ is fixed, the assumed $L_2$-continuity of the $L$ process implies that the right-hand side vanishes as $r \to 0$.  Therefore, we have that $\widetilde L_r^m \to L(s^\star)^{\otimes m}$ in probability and, hence, in distribution.  

For the $E$ part in the above representation, locally iid implies 
\[ \widetilde E_r^m \to \widetilde E_0^m, \quad \text{in distribution, as $r \to 0$} \]
where $\widetilde E_0^m$ is an iid vector.  This and the continuous mapping theorem gives 
\[ \widetilde T_r^m \to \psi\bigl( L(s^\star)^{\otimes m}, \widetilde E_0^m \bigr) \quad \text{in distribution, as $r \to 0$}, \]
where $\psi$ is being applied component-wise.  The right-hand side is clearly conditionally iid, given $L(s^\star)$, which implies exchangeability.  Finally, since this holds for all $m$ and for all $(u_1,\ldots,u_m)$, we know that the process $\widetilde T_r$ has a distributional limit, and that limit is an exchangeable process.

\subsubsection*{Proof of Proposition 2} 
\label{proof:sphere}

Let $\rho$ be a generic $3 \times 3$ rotation matrix that takes a location $s$ on the sphere $\calD$ to a new location $\rho s$, also on the sphere.  Write $\rho \prob$ for the distribution of the spatial process post-rotation by $\rho$, $(\rho S, X(\rho S), Y(\rho S))$.  From our stated distributional assumptions, namely, that $S$ is uniform on $\calD$ and $(X,Y)$ is isotropic and stationary, it follows that the joint distribution is invariant to rotations of the sphere, i.e., $\rho \prob = \prob$.  Consequently, the conditional coverage probability function satisfies 
\begin{equation}
\label{eq:cvg1}
c(\rho s^\star \mid \alpha, n, \rho \prob) = c(\rho s^\star \mid \alpha, n, \prob). 
\end{equation}
Next, by the Kriging-based construction of the conformal prediction interval, it is similarly easy to see that the coverage event is invariant to rotations too.  That is, if $\rho Z^n$ is the data after rotating the spatial locations by $\rho$, then 
\begin{align*}
\Gamma^\alpha(\rho Z^n; & \, \rho S_{n+1}, X(\rho S_{n+1})) \ni Y(\rho S_{n+1}) \\
& \iff \Gamma^\alpha(Z^n; S_{n+1}, X(S_{n+1})) \ni Y(S_{n+1}). 
\end{align*}
From here, it follows that the conditional coverage probability function satisfies
\begin{equation}
\label{eq:cvg2}
c(\rho s^\star \mid \alpha, n, \rho \prob) = c(s^\star \mid \alpha, n, \prob). 
\end{equation}
Putting together the equalities in \eqref{eq:cvg1} and \eqref{eq:cvg2}, we conclude that $s^\star \mapsto c(s^\star \mid \alpha, n, \prob)$ is constant.  But the result in Theorem~\ref{thm:valid.global} applies in present case; therefore, if the marginal coverage probability exceeds $1-\alpha$, then the constant conditional coverage probability must exceed $1-\alpha$ too, which completes the proof.

\subsubsection*{Proof of Theorem~\ref{thm:valid.local}}
\label{proof:valid.local}

Let $T(s) = (X(s), Y(s))$ be the joint response-covariate process.  We have $m$ spatial locations in a neighborhood of $s^\star$, which can be expressed as $s_i = s^\star + r u_i$, $i=1,\ldots,m$, where $u_i \in \calU$ and $r$ is the neighborhood's radius.  Set $s_{m+1} = s^\star$ and define $X_i = X(s_i)$ and $Y_i=Y(s_i)$, for $i=1,\ldots,m,m+1$.  Let $\widetilde T_r^{m+1}$ denote the collection of all $m+1$ triples $(s_i, X_i, Y_i)$, including the $(m+1)^\text{st}$ entry that corresponds to the values at $s^\star$.  The non-conformity scores, $\delta_1,\ldots,\delta_{m+1}$, are functions of $\widetilde T_r^{m+1}$, and we will collect these into an $(m+1)$-vector $\delta_r^{m+1}$, whose dependence on the radius $r$ is now being made explicit.  Then 
\[ \Gamma_{s^\star}^\alpha(Z^m; X_{m+1}) \ni Y_{m+1} \iff \rank(m+1; \delta_r^{m+1}) > \lfloor (m+1) \alpha \rfloor, \]
where $\rank(m+1; \delta_r^{m+1})$ denotes the rank of the $(m+1)^\text{st}$ entry within the collection $\delta_r^{m+1}$.  By Proposition~\ref{thm:local.exchangeable}, we have that $\widetilde T_r^{m+1}$ converges in distribution, as $n \to \infty$ or, equivalently, as $r \to 0$, to an exchangeable $\widetilde T_0^{m+1}$.  Since the non-conformity score vector, $\delta_r^{m+1}$, is a continuous function of $\widetilde T_r^{m+1}$, it follows from the continuous mapping theorem that, as $r \to 0$, $\delta_r^{m+1}$ converges in distribution to $\delta_0^{m+1}$, say, which is just the non-conformity measure applied to the limit $\widetilde T_0^{m+1}$.  Now, turning to the ranks, each permutation of $1,\ldots,m+1$ is a possible value of the ranks, and it corresponds to an event $A$ in the space of $\delta_0^{m+1}$.  The boundary of that event consists of cases in which there are ties in $\delta_0^{m+1}$.  By our continuity assumptions, this boundary has probability 0, which makes $A$ a continuity set.  Then it follows from the Portmanteau lemma \citep[e.g.,][Lemma~2.2]{vaart1998} that 
\[ \prob(\delta_r^{m+1} \in A) \to \prob(\delta_0^{m+1} \in A), \quad r \to 0. \]
This holds for the $A$ corresponding to each configurations of the ranks which, in turn, implies 
\[ \rank(m+1; \delta_r^{m+1}) \to \rank(m+1; \delta_0^{m+1}) \quad \text{in distribution, as $r \to 0$}. \]
Since the limit $\widetilde T_0^{m+1}$ is exchangeable, the structure of the non-conformity measure implies that $\delta_0^{m+1}$ is also exchangeable.  Therefore, $\rank(m+1; \delta_r^{m+1})$ converges in distribution to $U \sim \text{Unif}(\{1,\ldots,m,m+1\})$, so
\[ \lim_{n \to \infty} \prob^{m+1}\{ \Gamma_{s^\star}^\alpha(Z^m; X_{m+1}) \ni Y_{m+1} \} = \prob\{ U > \lfloor (m+1)\alpha \rfloor\}. \]
Finally, the probability on the right-hand side is $1-\alpha + O(m^{-1})$, so that the limiting coverage probability is approximately $1-\alpha$ when $m$ is large.

\subsection*{Theoretical efficiency of GSCP}
\label{S:efficiency}

\subsubsection*{Setup and statement of the result}

Here we establish an asymptotic, theoretical efficiency result for (a version of) our global spatial conformal prediction framework.  Since the focus is on asymptotics, we can consider a simplified version, a variation on the so-called {\em split conformal} prediction strategy \citep[e.g.,][]{lei2018distribution, lei2014distribution, vovk2005algorithmic}.  Without any real loss of generality, assume we have a sample of $2n$ location-covariate-response triples, i.e., $Z_i = (S_i, X_i, Y_i)$ for 
\[ i \in \calI_1^n = \{1,\ldots,n\} \quad \text{and} \quad i \in \calI_2^n = \{n+1,\ldots,2n\}. \]
We will use the $\calI_2^n$ data to estimate the prediction mean and standard deviation functions, and the $\calI_1^n$ data to do conformal prediction.  This simplifies the analysis considerably.  To see this, let $(\hat\mu_n, \hat\sigma_n)$ be the prediction mean and standard deviation functions estimated based on data of size $n$ in $\calI_2^n$.  If we are going to apply the non-conformity measure only to data in $\calI_1^n$, then it is just applying a fixed function to each $Z_i$,  
\[ \Delta_n(Z_i) = \Bigl| \frac{Y_i - \hat\mu_n(S_i,X_i)}{\hat\sigma_n(S_i,X_i)} \Bigr|, \quad i \in \calI_1^n, \]
and we do not need a subscript ``$(i)$'' to indicate that observation $i$ was excluded in training the prediction rule.  That is, in split conformal, the data (in $\calI_2^n$) for training the prediction rule does not mix with the data (in $\calI_1^n$) for constructing the prediction interval as it does the in full conformal prediction algorithm.  

Under this setup, the (split) conformal prediction interval can take a relatively simple form.  Let $\hat q_{n,\alpha}$ denote the upper $\alpha/2$ quantile of $\{\Delta_n(Z_i): i \in \calI_1^n\}$.  If we assume symmetry in the distribution of the signed $\Delta_n(Z_i)$'s---see below---then a $100(1-\alpha)$\% (split) conformal prediction interval for a new response value $\tilde Y$ associated with a new location-covariate pair $(\tilde S,\tilde X)$, is 
\[ \hat\mu_n(\tilde S, \tilde X) \pm \hat\sigma_n(S,X) \, \hat q_{n,\alpha}. \]
Clearly, the width of this interval is $2 \hat\sigma_n(\tilde S, \tilde X) \, \hat q_{n,\alpha}$.  Our goal is to show that this is nearly the width of the ``optimal'' prediction interval, as $n \to \infty$, and hence that (split) conformal is not only valid but asymptotically efficient.  

For the analysis that follows, we assume that the working model presented in Section~\ref{s:back:spatial} is correct, i.e., that the structural mean is linear in covariates $X$, that there is an underlying Gaussian process $\theta$ with isotropic \Matern\ covariance function depending on parameters $(\sigma^2, \phi, \kappa)$, and a Gaussian nugget with variance $\tau^2$.  This assumption is necessary because very little is known about the convergence properties of the Kriging estimates under misspecification; in other words, we would not be able to characterize an ``optimal'' prediction interval under weaker assumptions.  

When we consider ``$n \to \infty$,'' we are assuming an infill asymptotics regime wherein the range of spatial locations $\calD$ is a fixed, compact set but the number $n$ (or $2n$) of observed triples $(S,X,Y)$ is going to infinity.  Also assume that the range of $X$ is compact.  Let $\mu^\star$ and $\sigma^\star$ be the true prediction mean and standard deviation functions under the working model.  More specifically, $\mu^\star = \mu_\theta^\star$ depends on the Gaussian process $\theta$ as is given by 
\[ \mu^\star(s,x) = x^\top \beta^\star + \theta(s), \]
and $\sigma^\star(s,x) \equiv \tau^\star$ is a constant, the nugget standard deviation.  If all these starred quantities were known, then one could construct a prediction interval for a new response $\tilde Y$ associated with a new location-covariate pair $(\tilde S, \tilde X)$ as follows.  Define 
\[ \Delta^\star(\tilde Z) = \Bigl| \frac{\tilde Y - \mu^\star(\tilde S, \tilde X)}{\sigma^\star(\tilde S, \tilde X)} \Bigr|, \quad \tilde Z=(\tilde S, \tilde X, \tilde Y), \]
which is just the non-conformity measure applied to $\tilde Z$ using the true prediction mean and standard deviation functions.  Under the true working model, the standardized residual in the above display is standard normal.  Let $q_\alpha^\star$ denote the upper $\alpha/2$ quantile of $\text{Normal}(0,1)$.  Then the ``optimal'' prediction interval for $\tilde Y$ at $(\tilde S, \tilde X)$ would be 
\[ \mu^\star(\tilde S, \tilde X) \pm q_\alpha^\star \, \sigma^\star(\tilde S, \tilde X). \]
Clearly, the length of this optimal prediction interval is $2 q_\alpha^\star \, \sigma^\star(\tilde S, \tilde X)$.  Define the absolute difference (ignoring the factor 2) between the widths of the split conformal and optimal prediction intervals as 
\[ \Xi_n(\tilde S, \tilde X) = |\hat q_{n,\alpha} \, \hat\sigma_n(\tilde S, \tilde X) - q_\alpha^\star \, \sigma^\star(\tilde S, \tilde X)|, \]
which is a function of the data $\{Z_i: i \in \calI_2^n\}$ and the new location-covariate pair $(\tilde S, \tilde X)$.  Our claim is that $\Xi_n(\tilde S, \tilde X)$ vanishes in probability at a certain rate as $n \to \infty$.  This is similar to what \citet{lei2018distribution} prove in their Theorem~6. 

Finally, given a vanishing sequence $\eta_n \to 0$, define the event $\mathcal{A}_n$ as 
\[ \mathcal{A}_n = \Bigl\{ \Bigl\| \frac{\hat\mu_n - \mu^\star}{\hat\sigma_n} \Bigr\|_\infty \vee \Bigl\| \frac{\hat\sigma_n - \sigma^\star}{\hat\sigma_n} \Bigr\|_\infty \leq \eta_n \Bigr\}, \]
where the mean and standard deviations are functions of location-covariate pairs in a compact domain, and $\|\cdot\|_\infty$ is the supremum norm.  Note that the best possible rate is $\eta_n \sim c_n n^{-1/2}$, where $c_n \to \infty$ arbitrarily slowly.  

\begin{thm}
\label{thm:efficiency}
If $\eta_n \to 0$ is such that $\prob(\mathcal{A}_n) \to 1$ as $n \to \infty$, 
then, 
\[ \Xi_n(\tilde S, \tilde X) = O_\prob\{\eta_n(\log n)^{1/2}\}, \quad n \to \infty. \]
\end{thm}

\begin{proof}
See below.
\end{proof}

Here we make two remarks about the theorem---one about the condition and the other about the conclusion.  First, \citet{lei2018distribution} consider an event similar to our $\mathcal{A}_n$.  Ours appears more complicated because we work with standardized residuals and, therefore, rely on both prediction mean and standard deviation functions, whereas they use unstandardized residuals and only analyze a prediction mean function.  It is well known that under infill asymptotics, the estimators of the spatial correlation parameters generally are not consistent \citep[e.g.,][]{zhang2004}, but the prediction mean and standard deviation functions can be estimated consistently \citep[e.g.,][and references therein]{tang2019identifiability}.  Moreover, under certain regularity conditions as described in \cite{tang2019identifiability}, we expect that the rate $\eta_n$ in the above theorem would be $n^{-r}$, where $r=\frac{\kappa}{2(\kappa + 1)} < 1/2$ and $\kappa > 0$ is the $\Matern$ covariance function's smoothness parameter.  Whether the assumption of Theorem~\ref{thm:efficiency} holds for this $\eta_n$ has yet to be confirmed.  

Second, the logarithmic term in the rate is a result of the proof technique and so likely is not needed.  Therefore, we conjecture that the result could be improved to $\Xi_n(\tilde S, \tilde X) = O_\prob(\eta_n)$, but this slight improvement in the rate would require a different or substantially more involved proof.  So we leave establishing the improved convergence rate result as a topic of future research.

\subsubsection*{Proof of Theorem~\ref{thm:efficiency}}

To start, write 
\[ \Xi_n(\tilde S, \tilde X) \leq \hat\sigma_n(\tilde S, \tilde X) \, | \hat q_{n,\alpha} - q_\alpha^\star| + |\hat\sigma_n(\tilde S, \tilde X) - \sigma^\star(\tilde S, \tilde X)| \, q_\alpha^\star. \]
The condition $\prob(\mathcal{A}_n) \to 1$ has two important implications: first,  $\hat\sigma_n(\tilde S, \tilde X) = O_\prob(1)$; second, $|\hat\sigma_n(\tilde S, \tilde X) - \sigma^\star(\tilde S, \tilde X)| = O_\prob(\eta_n)$.  Therefore, if we can show that $|\hat q_{n,\alpha} - q_\alpha^\star| = O_\prob\{\eta (\log n)^{1/2}\}$, then we are done.  

Let $\hat q_{n,\alpha}^\star$ denote the upper $\alpha/2$ quantile of the empirical distribution of 
\[ \{\Delta^\star(Z_i): i \in \calI_1^n\}. \]
Then we have 
\[ |\hat q_{n,\alpha} - q_\alpha^\star| \leq |\hat q_{n,\alpha} - \hat q_{n,\alpha}^\star| + |\hat q_{n,\alpha}^\star - q_\alpha^\star|. \]
Then for a generic sequence $\zeta_n$,
\begin{align}
\prob\{|\hat q_{n,\alpha} &- q_\alpha^\star| > \zeta_n\} \notag \\
& \leq \prob\{|\hat q_{n,\alpha} - \hat q_{n,\alpha}^\star| + |\hat q_{n,\alpha}^\star - q_\alpha^\star| > \zeta_n\} \notag \\
& \leq \prob\{|\hat q_{n,\alpha} - \hat q_{n,\alpha}^\star| > \tfrac12 \zeta_n\} + \prob\{|\hat q_{n,\alpha}^\star - q_\alpha^\star| > \tfrac12 \zeta_n\}. \label{eq:qbound}
\end{align}
It remains to show that, for $\zeta_n \sim \eta_n (\log n)^{1/2}$, both terms in \eqref{eq:qbound} are $o(1)$.  

Start with the second term in \eqref{eq:qbound}.  Given the pair of processes $(X,\theta)$, the $Z_i$'s are iid, so, conditionally, the empirical quantile $\hat q_{n,\alpha}^\star$ would converge to the standard normal quantile $q_\alpha^\star$ at a $n^{-1/2}$ rate.  Indeed,
\[ \prob\{|\hat q_{n,\alpha}^\star - q_\alpha^\star| > \tfrac12 \zeta_n\} = \E \bigl[ \prob\{|\hat q_{n,\alpha}^\star - q_\alpha^\star| > \tfrac12 \zeta_n \mid X, \theta\} \bigr], \]
where the right hand side is the expectation of a version of the conditional distribution, given the processes $(X,\theta)$.  For almost all $(X,\theta)$, the conditional probability vanishes since $\zeta_n \gg n^{-1/2}$.  The conditional probability is also bounded, so the dominated convergence theorem implies that its expectation, $\prob\{|\hat q_{n,\alpha}^\star - q_\alpha^\star| > \tfrac12 \zeta_n\}$, vanishes as well.  

Now take the first term in \eqref{eq:qbound}.  To start, for any $z=(s,x,y)$, it is easy to check that 
\[ |\Delta_n(z) - \Delta^\star(z)| \leq \Bigl| \frac{\hat\mu_n(s,x) - \mu^\star(s,x)}{\hat\sigma_n(s,x)} \Bigr| + \Bigl| \frac{\sigma^\star(s,x)}{\hat\sigma_n(s,x)} - 1 \Bigr| \Delta^\star(z). \]
Define the event 
\[ \mathcal{B}_n = \Bigl\{ \max_{i \in \calI_1^ n} \Delta^\star(Z_i) \leq c(\log n)^{1/2} \Bigr\}, \]
for a suitable constant $c > 0$.  If we restrict ourselves to samples in the event $\mathcal{A}_n \cap \mathcal{B}_n$, then, for any $z$, 
\[ |\Delta_n(z) - \Delta^\star(z)| \leq \eta_n\{1 + \Delta^\star(z)\} \lesssim \eta_n (\log n)^{1/2}. \]
As in \citet{lei2018distribution}, if all the differences are uniformly bounded, then the corresponding difference in sample quantiles is also bounded.  That is, the above display implies $|\hat q_{n,\alpha} - \hat q_{n,\alpha}^\star| \lesssim \eta_n (\log n)^{1/2}$.  Therefore, 
\[ 1-\prob\{|\hat q_{n,\alpha} - \hat q_{n,\alpha}^\star| > \tfrac12 \zeta_n\} \geq \prob(\mathcal{A}_n \cap \mathcal{B}_n). \]
We know that $\prob(\mathcal{A}_n) \to 1$ so it suffices to show that $\prob(\mathcal{B}_n) \to 1$ too.  Do the same conditioning operation as above:
\[ \prob(\mathcal{B}_n^c) = \E\{ \prob(\mathcal{B}_n^c \mid X,\theta) \}. \]
Given the processes $(X,\theta)$, the the standardized residuals in $\Delta^\star(Z_i)$ are iid $\text{Normal}(0,1)$ and, hence, we can apply the standard Gaussian maximal inequality to conclude that $\prob(\mathcal{B}_n^c \mid X,\theta) \to 0$ for almost all $(X,\theta)$.  Then the dominated convergence theorem again gives $\prob(\mathcal{B}_n) \to 1$, from which we can conclude that $\prob\{|\hat q_{n,\alpha} - \hat q_{n,\alpha}^\star| > \tfrac12 \zeta_n\} \to 0$, hence completing the proof.

\subsection*{Additional simulation results}
\label{s:A3}

This appendix presents some additional simulation results. First, Figure \ref{f:sim:overview} shows a realization from each scenario to help with visualizing the various spatial correlation patterns. Second, we compare two variations of the proposed (local) spatial conformal prediction strategy: a split version where the working model parameters are estimated based on one half of the data, and then conformal with the estimated parameters applied on the other half, and a plug-in version where the parameters are estimated via and conformal prediction is applied to the entire data set.  Finally, we evaluate the proposed methods in applications where the data sets include covariates, and study sensitivity to the choice of tuning parameters and the base model.

\begin{figure}
\caption{One realization each from simulation scenarios with $N=40.$}\label{f:sim:overview}
\centering
\includegraphics[width=0.6\textwidth]{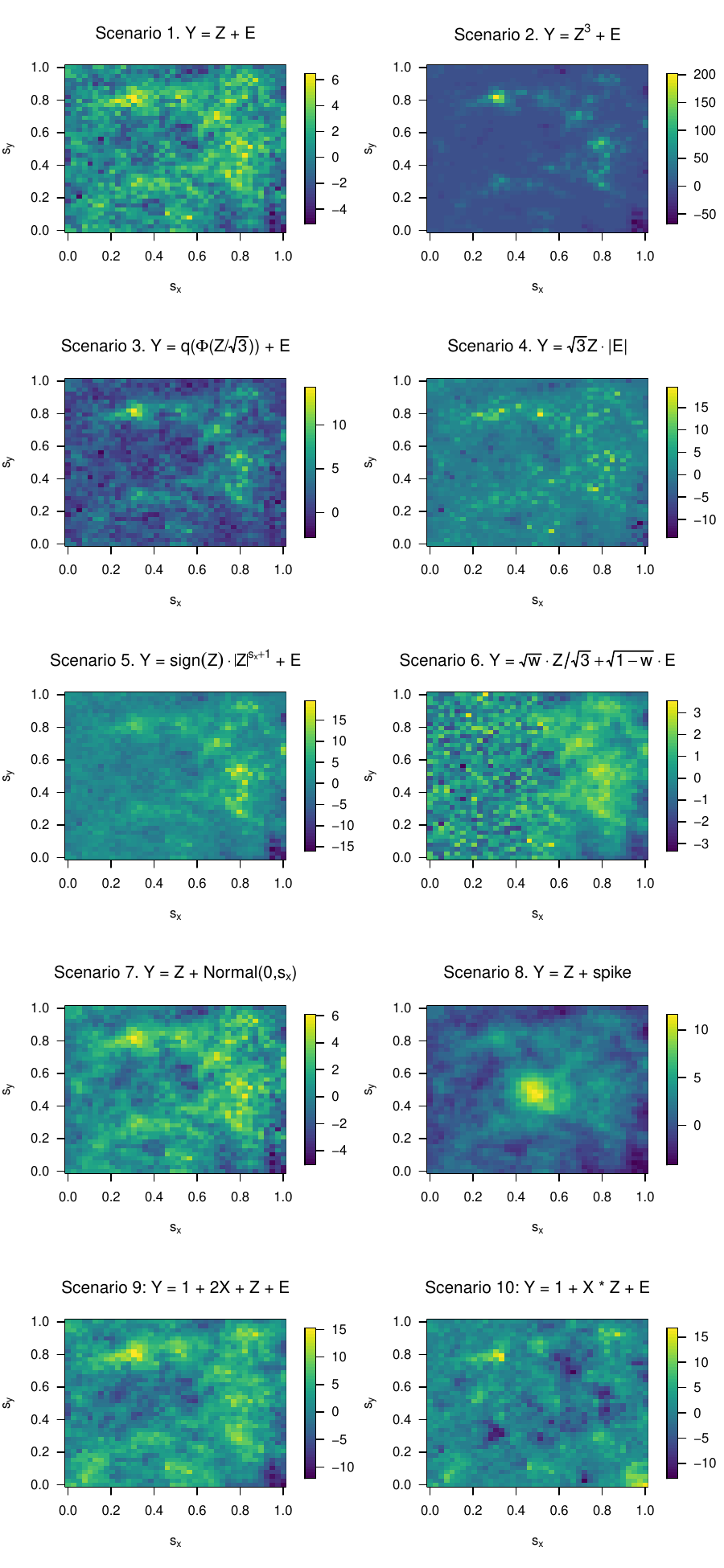}
\end{figure}

\subsubsection*{Split versus plug-in conformal prediction results}

In split conformal prediction \citep[e.g.,][]{lei2018distribution, lei2014distribution, vovk2005algorithmic}, the training data are split (equally and randomly) into two groups. The first group is used to estimate the model parameters $\Theta$ and is then discarded.  The algorithm then proceeds with only the remaining training observations for conformal prediction, including evaluating the base model (given parameters estimates) and plausibilities.  Table \ref{t:sim_conformal} compares the results using the settings defined in the main document with and without split conformal.  Although split conformal maintains good coverage, the intervals are a little wider on average and thus the interval scores are higher.

\begin{table}
\caption{\label{t:sim_conformal}Performance comparison for simulation scenarios (``Scen'') with (``- split'') and without split conformal. The metrics are the empirical coverage of 90\% prediction intervals (``Cov90''), the width of prediction intervals (``Width'') and the interval score (``IntScore''), each averaged over location and dataset. The methods are global (GSCP) and local (LSCP) conformal prediction. }
\begin{tabular}{ll|ccc|ccc}
 & & \multicolumn{3}{c|}{$N=20$} & \multicolumn{3}{c}{$N=40$}   \\
 Scen & Method & Cov90 & Width & IntScore & Cov90 & Width & IntScore \\ 
  \hline
\ref{scen:stationary}
& GSCP - split & 0.905 & 5.11 & 6.32 & 0.893 & 4.19 & 5.35 \\       
& GSCP         & 0.906 & 4.67 & 5.78 & 0.897 & 4.00 & 5.06 \\               
& LSCP - split & 0.870 & 4.92 & 6.81 & 0.883 & 4.16 & 5.47 \\      
& LSCP         & 0.890 & 4.57 & 5.95 & 0.891 & 3.99 & 5.12 \\   
&&&&&&&\vspace{-6pt}\\
\ref{scen:cubic}
& GSCP - split & 0.900 & 39.32 & 78.75 & 0.897 & 26.06 & 51.36 \\   
& GSCP & 0.895 & 33.62 & 67.16 & 0.897 & 22.12 & 43.80 \\           
& LSCP - split & 0.885 & 36.76 & 72.26 & 0.907 & 25.28 & 43.26 \\   
& LSCP & 0.896 & 31.05 & 58.37 & 0.910 & 21.74 & 36.47 \\           
&&&&&&&\vspace{-6pt}\\
\ref{scen:skewed}
& GSCP - split & 0.905 & 5.31 & 7.19 & 0.894 & 4.25 & 5.61 \\       
& GSCP & 0.908 & 4.79 & 6.32 & 0.895 & 4.05 & 5.27 \\               
& LSCP - split & 0.869 & 5.17 & 7.42 & 0.889 & 4.23 & 5.62 \\       
& LSCP & 0.893 & 4.65 & 6.27 & 0.893 & 4.03 & 5.24 \\               
&&&&&&&\vspace{-6pt}\\
\ref{scen:prod} 
& GSCP - split & 0.906 & 8.03 & 12.27 & 0.896 & 6.51 & 10.62 \\     
& GSCP & 0.902 & 7.06 & 11.06 & 0.895 & 6.25 & 10.25 \\             
& LSCP - split & 0.873 & 7.58 & 12.79 & 0.892 & 6.30 & 10.32 \\     
& LSCP & 0.892 & 6.84 & 11.23 & 0.895 & 6.08 &  9.74 \\             
&&&&&&&\vspace{-6pt}\\
\ref{scen:gfun} 
& GSCP - split & 0.900 & 7.47 & 11.02 & 0.895 & 5.44 & 7.79 \\      
& GSCP & 0.900 & 6.51 &  9.40 & 0.898 & 5.03 & 7.11 \\              
& LSCP - split & 0.869 & 7.21 & 11.07 & 0.888 & 5.45 & 7.50 \\      
& LSCP & 0.887 & 6.36 &  9.06 & 0.892 & 5.05 & 6.75 \\              
&&&&&&&\vspace{-6pt}\\
\ref{scen:eastwest} 
& GSCP - split & 0.901 & 2.99 & 3.83 & 0.897 & 2.67 & 3.64 \\       
& GSCP & 0.894 & 2.78 & 3.69 & 0.896 & 2.63 & 3.60 \\               
& LSCP - split & 0.861 & 2.79 & 3.81 & 0.889 & 2.47 & 3.23 \\       
& LSCP & 0.878 & 2.66 & 3.47 & 0.895 & 2.38 & 3.06 \\               
&&&&&&&\vspace{-6pt}\\
\ref{scen:nugget} 
& GSCP - split & 0.904 & 4.14 & 5.14 & 0.897 & 3.02 & 4.01 \\       
& GSCP & 0.905 & 3.63 & 4.55 & 0.896 & 2.77 & 3.71 \\              
& LSCP - split & 0.870 & 4.00 & 5.49 & 0.890 & 2.94 & 3.86 \\       
& LSCP & 0.888 & 3.53 & 4.59 & 0.896 & 2.70 & 3.46 \\               
&&&&&&&\vspace{-6pt}\\
\ref{scen:spike} 
& GSCP - split & 0.911 & 3.85 & 4.79 & 0.897 & 2.25 & 2.89 \\      
& GSCP & 0.906 & 3.04 & 3.74 & 0.899 & 1.93 & 2.45 \\               
& LSCP - split & 0.872 & 3.72 & 5.01 & 0.888 & 2.23 & 2.92 \\       
& LSCP & 0.880 & 3.00 & 3.91 & 0.895 & 1.92 & 2.46 \\               
\end{tabular}
\end{table}

\subsubsection*{Simulations with covariates}

To study the performance of the method with covariates, we consider two additional scenarios
\begin{enumerate}[label=\arabic*, labelsep=0pt]
    \item[9] \label{scen:addcov}. $Y(s) = 1 + 2X(s) + Z(s) + E(s)$; 
    \item[10] \label{scen:prodcov}. $Y(s) = 1 + X(s)Z(s) + E(s)$. 
\end{enumerate}
where the covariate process $X(s)$ is a mean-zero stationary Gaussian process processes with $\Matern$ covariance function with variance $\sigma^2=3$, range $\phi=0.1$ and smoothness $\kappa = 0.7$. Table \ref{t:sim_withCov} summarizes the performance of the methods with and without using the covariate in the analysis.  In Scenario 9, the covariate enters the model additively as assumed by all methods and thus including the covariate reduces the interval score.  In Scenario 10, the covariate is not additive and the Kriging models are misspecified. However, all conformal methods maintain coverage near the nominal level when $N=40$.

\begin{table}
\caption{\label{t:sim_withCov}Performance comparison for simulation scenarios with a covariate. The metrics are the empirical coverage of 90\% prediction intervals (``Cov90''), the width of prediction intervals (``Width'') and the interval score (``IntScore''), each averaged over location and dataset. The methods are global (GSCP) and local (LSCP) conformal prediction, stationary and Gaussian Kriging (``Kriging'').  Methods that use the covariate are denoted ``(X)''. }
\centering
\begin{tabular}{ll|ccc|ccc}
 & & \multicolumn{3}{c|}{$N=20$} & \multicolumn{3}{c}{$N=40$}   \\
 Scen & Method & Cov90 & Width & IntScore & Cov90 & Width & IntScore \\ 
  \hline
9
& GSCP       & 0.893 & 7.53 & 9.61 & 0.899 & 5.54 & 6.88\\ 
& GSCP (X)    & 0.896 & 4.48 & 5.90 & 0.900 & 4.04 & 5.02\\ 
& LSCP       & 0.877 & 7.45 & 9.92 & 0.897 & 5.52 & 6.94\\ 
& LSCP (X)    & 0.885 & 4.62 & 6.10 & 0.896 & 4.02 & 5.06\\ 
& Kriging    & 0.909 & 7.89 & 9.59 & 0.888 & 5.39 & 6.96\\ 
& Kriging (X) & 0.903 & 4.76 & 5.89 & 0.889 & 3.91 & 5.03\\ 
&&&&&&&\vspace{-6pt}\\
10
& GSCP       & 0.898 & 7.73 & 11.16 & 0.899 & 5.68 & 7.68\\ 
& GSCP (X)    & 0.898 & 7.65 & 11.09 & 0.899 & 5.67 & 7.62\\ 
& LSCP       & 0.885 & 7.60 & 10.90 & 0.897 & 5.67 & 7.56\\ 
& LSCP (X)    & 0.886 & 7.52 & 10.90 & 0.897 & 5.65 & 7.51\\ 
& Kriging    & 0.910 & 8.19 & 11.21 & 0.885 & 5.49 & 7.79\\ 
& Kriging (X) & 0.909 & 8.09 & 11.13 & 0.880 & 5.44 & 7.74\\ 
\end{tabular}
\end{table}

\subsubsection*{Sensitivity analysis}\label{s:sim:sensitivity}

Here we study the proposed method's sensitivity to tuning parameter choice and to the initial $\Matern$ covariance parameter estimates. Because these results are not addressing parameter estimation, in this subsection we estimate the model parameters using the entire dataset (or treat them as fixed, as described in the next paragraph) and then perform leave-one-out cross validation.  Figure \ref{f:sensitivity} gives the average interval score of LSCP with different tuning parameter $\eta$ for Scenarios \ref{scen:gfun} and \ref{scen:eastwest}, along with the average interval score of  GSCP, Kriging, and laGP. We find that there is a V-shaped relationship between resulting interval score and tuning parameter $\eta$. The average performance is the best when $\eta=0.2$, and LSCP outperforms the other methods regardless the selection of $\eta$. To achieve the best performance, in practice, we would select $\eta$ using cross validation. 

\begin{figure}
\centering 
\includegraphics[width=1\textwidth]{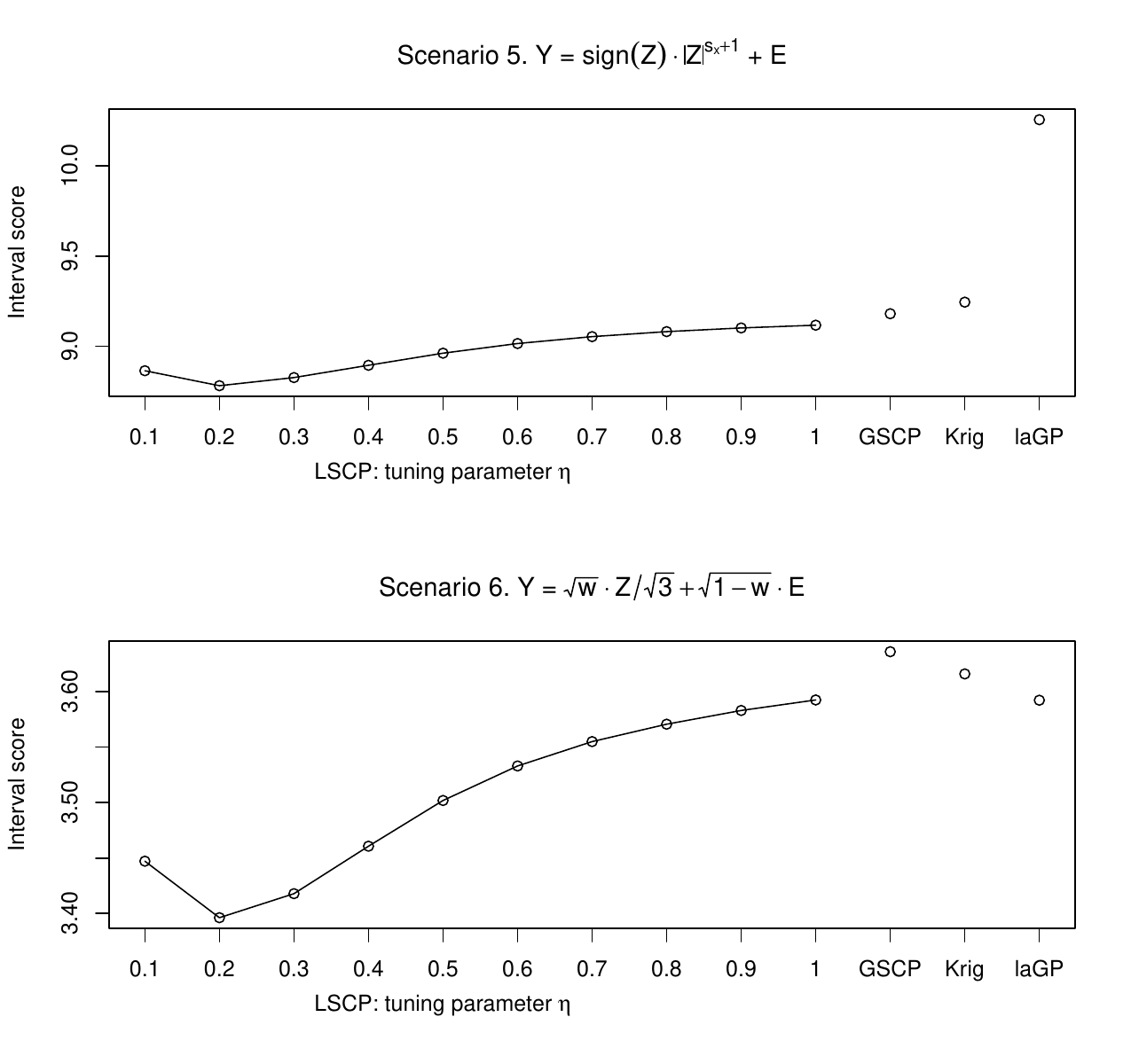}
\caption{Sensitivity analysis for simulation Scenarios \ref{scen:gfun} and \ref{scen:eastwest}.  The interval score for the LSCP method is computed for several kernel bandwidth parameters, $\eta$, and compared to other methods.}
\label{f:sensitivity}
\end{figure}

The proposed spatial conformal prediction methods rely on estimation of the $\Matern$ covariance parameter, $\Theta$. To test for sensitivity to this step, we repeat simulation Scenario 1 since in this case the optimal $\Theta$ is known to be the true $\Matern$ covariance parameters used to generate the data. When generating prediction intervals, rather than estimating the $\Matern$ covariance parameter $\hat\Theta$ from the simulated data, we use the values of $\hat\Theta$ listed in the first four columns of Table \ref{t:sensitivity_thetaHat} for conformal prediction and Kriging, where each of the four parameters is modified by plus or minus 50\% of the true value. We use the GSCP method, $\alpha=0.1$ and $N=20$ for this sensitivity analysis. The corresponding average coverage and average width of the conformal prediction intervals and the Kriging intervals are listed in the last four columns.

\begin{table}
\caption{\label{t:sensitivity_thetaHat} Sensitivity analysis for the simulation study. The first four columns represent the $\Matern$ covariance parameter ($\Theta$) used to generate the conformal and Kriging prediction intervals; the first row gives the true value. The last four columns give the average coverage and interval width of the 90\% intervals produced by the conformal and Kriging methods when using the corresponding $\Theta$.}
\centering
\begin{tabular}{cccc|cc|cc}
Nugget& Partial sill & Range  & Smoothness & \multicolumn{2}{c|}{GSCP} & \multicolumn{2}{c}{Kriging} \\
($\tau^2$) &($\sigma^2$) &($\phi$)& ($\kappa$) & Cov90 & Width & Cov90 & Width \\  \hline
1.0 & 3.0 & 0.10 & 0.70 & 89.69\% & 4.57 & 90.05\% & 4.60 \\ 
1.5 & 3.0 & 0.10 & 0.70 & 89.69\% & 4.61 & 93.73\% & 5.24 \\ 
0.5 & 3.0 & 0.10 & 0.70 & 89.70\% & 4.59 & 82.76\% & 3.84 \\ 
1.0 & 4.5 & 0.10 & 0.70 & 89.68\% & 4.63 & 92.64\% & 5.04 \\ 
1.0 & 1.5 & 0.10 & 0.70 & 89.68\% & 4.63 & 85.25\% & 4.09 \\ 
1.0 & 3.0 & 0.15 & 0.70 & 89.68\% & 4.60 & 86.71\% & 4.21 \\ 
1.0 & 3.0 & 0.05 & 0.70 & 89.67\% & 4.61 & 95.13\% & 5.54 \\ 
1.0 & 3.0 & 0.10 & 1.05 & 89.69\% & 4.58 & 85.56\% & 4.08 \\ 
1.0 & 3.0 & 0.10 & 0.35 & 89.70\% & 4.63 & 94.98\% & 5.53
\end{tabular}
\end{table}

While using the true $\Theta$ leads to slightly more narrow intervals, the average coverage of the conformal prediction intervals are between 89.67\% and 89.70\% across $\Theta$.  In contrast, the coverage of the Kriging intervals range from 82.76\% to 95.13\% across $\Theta$.  Therefore, the conformal methods are more robust to the estimation of $\Matern$ covariance parameters than the parametric Kriging model.

\section*{Acknowledgments}

The authors thank the reviewers for their thoughtful and critical comments on a previous version of the manuscript. HM (DMS--1638521) and RM (DMS--1811802) are supported by the National Science Foundation. BJR is supported by the National Institutes of Health (R01ES031651 and R01ES027892), the National Science Foundation (DMS--2152887) and King Abdullah University of Science and Technology (3800.2).

\bibliographystyle{apalike} 
\bibliography{refs}

\begin{thebibliography}{}

\bibitem[Barber et~al., 2019]{barber2019limits}
Barber, R.~F., Cand{\`e}s, E.~J., Ramdas, A., and Tibshirani, R.~J. (2019).
\newblock The limits of distribution-free conditional predictive inference.
\newblock {\em arXiv preprint arXiv:1903.04684}.

\bibitem[Cella and Martin, 2022]{cella2020valid}
Cella, L. and Martin, R. (2022).
\newblock Validity, consonant plausibility measures, and conformal prediction.
\newblock {\em International Journal of Approximate Reasoning}, 141:110--130.

\bibitem[Chernozukov et~al., 2021]{chernozukov.dconformal}
Chernozukov, V., W{\"u}thrich, K., and Zhu, Y. (2021).
\newblock Distributional conformal prediction.
\newblock {\em Proceedings of the National Academy of Sciences},
  118(48):e2107794118.

\bibitem[Cook et~al., 2013]{cook2013nasa}
Cook, B., Nelson, R., Middleton, E., Morton, D., McCorkel, J., Masek, J.,
  Ranson, K., Ly, V., Montesano, P., et~al. (2013).
\newblock Nasa goddard’s lidar, hyperspectral and thermal (g-liht) airborne
  imager.
\newblock {\em Remote Sensing}, 5(8):4045--4066.

\bibitem[Cressie, 1992]{cressie1992statistics}
Cressie, N. (1992).
\newblock Statistics for spatial data.
\newblock {\em Terra Nova}, 4(5):613--617.

\bibitem[Datta et~al., 2016]{datta2016hierarchical}
Datta, A., Banerjee, S., Finley, A.~O., and Gelfand, A.~E. (2016).
\newblock {Hierarchical nearest-neighbor Gaussian process models for large
  geostatistical datasets}.
\newblock {\em Journal of the American Statistical Association},
  111(514):800--812.

\bibitem[Diggle et~al., 1998]{diggle1998model}
Diggle, P.~J., Tawn, J.~A., and Moyeed, R.~A. (1998).
\newblock Model-based geostatistics.
\newblock {\em Journal of the Royal Statistical Society: Series C (Applied
  Statistics)}, 47(3):299--350.

\bibitem[Duan et~al., 2007]{duan2007generalized}
Duan, J.~A., Guindani, M., and Gelfand, A.~E. (2007).
\newblock {Generalized spatial Dirichlet process models}.
\newblock {\em Biometrika}, 94(4):809--825.

\bibitem[Finley et~al., 2017]{finley2017spnngp}
Finley, A., Datta, A., and Banerjee, S. (2017).
\newblock {spNNGP: Spatial regression models for large datasets using nearest
  neighbor Gaussian processes}.
\newblock {\em R package version 0.1.0}, 1.

\bibitem[Franchi et~al., 2018]{franchi2018supervised}
Franchi, G., Yao, A., and Kolb, A. (2018).
\newblock Supervised deep kriging for single-image super-resolution.
\newblock In {\em German Conference on Pattern Recognition}, pages 638--649.
  Springer.

\bibitem[Fuglstad et~al., 2015]{fuglstad2015does}
Fuglstad, G.~A., Simpson, D., Lindgren, F., and Rue, H. (2015).
\newblock Does non-stationary spatial data always require non-stationary random
  fields?
\newblock {\em Spatial Statistics}, 14:505--531.

\bibitem[Gelfand et~al., 2005]{gelfand2005bayesian}
Gelfand, A.~E., Kottas, A., and MacEachern, S.~N. (2005).
\newblock {Bayesian nonparametric spatial modeling with Dirichlet process
  mixing}.
\newblock {\em {Journal of the American Statistical Association}},
  100(471):1021--1035.

\bibitem[Gneiting and Raftery, 2007]{gneiting2007strictly}
Gneiting, T. and Raftery, A.~E. (2007).
\newblock Strictly proper scoring rules, prediction, and estimation.
\newblock {\em Journal of the American Statistical Association},
  102(477):359--378.

\bibitem[Gramacy and Apley, 2015]{gramacy2015local}
Gramacy, R.~B. and Apley, D.~W. (2015).
\newblock Local {G}aussian process approximation for large computer
  experiments.
\newblock {\em Journal of Computational and Graphical Statistics},
  24(2):561--578.

\bibitem[Guan, 2019]{guan2019conformal}
Guan, L. (2019).
\newblock Conformal prediction with localization.
\newblock {\em arXiv preprint arXiv:1908.08558}.

\bibitem[Heaton et~al., 2019]{heaton2019case}
Heaton, M.~J., Datta, A., Finley, A.~O., Furrer, R., Guinness, J., Guhaniyogi,
  R., Gerber, F., Gramacy, R.~B., Hammerling, D., Katzfuss, M., Lindgren, F.,
  Nychka, D.~W., Sun, F., and Zammit-Mangion, A. (2019).
\newblock A case study competition among methods for analyzing large spatial
  data.
\newblock {\em Journal of Agricultural, Biological and Environmental
  Statistics}, 24(3):398--425.

\bibitem[Hengl et~al., 2018]{hengl2018random}
Hengl, T., Nussbaum, M., Wright, M.~N., Heuvelink, G.~B., and Gr{\"a}ler, B.
  (2018).
\newblock Random forest as a generic framework for predictive modeling of
  spatial and spatio-temporal variables.
\newblock {\em PeerJ}, 6:e5518.

\bibitem[Henley, 2012]{henley2012nonparametric}
Henley, S. (2012).
\newblock {\em Nonparametric Geostatistics}.
\newblock Springer Science \& Business Media.

\bibitem[Kim et~al., 2016a]{kim2016accurate}
Kim, J., Kwon~Lee, J., and Mu~Lee, K. (2016a).
\newblock Accurate image super-resolution using very deep convolutional
  networks.
\newblock In {\em Proceedings of the IEEE Conference on Computer Vision and
  Pattern Recognition}, pages 1646--1654.

\bibitem[Kim et~al., 2016b]{kim2016deeply}
Kim, J., Kwon~Lee, J., and Mu~Lee, K. (2016b).
\newblock Deeply-recursive convolutional network for image super-resolution.
\newblock In {\em Proceedings of the IEEE Conference on Computer Vision and
  Pattern Recognition}, pages 1637--1645.

\bibitem[Krupskii and Genton, 2019]{krupskii2019copula}
Krupskii, P. and Genton, M.~G. (2019).
\newblock A copula model for non-{G}aussian multivariate spatial data.
\newblock {\em Journal of Multivariate Analysis}, 169:264--277.

\bibitem[Lei et~al., 2018]{lei2018distribution}
Lei, J., G’Sell, M., Rinaldo, A., Tibshirani, R.~J., and Wasserman, L.
  (2018).
\newblock Distribution-free predictive inference for regression.
\newblock {\em Journal of the American Statistical Association},
  113(523):1094--1111.

\bibitem[Lei and Wasserman, 2014]{lei2014distribution}
Lei, J. and Wasserman, L. (2014).
\newblock Distribution-free prediction bands for non-parametric regression.
\newblock {\em Journal of the Royal Statistical Society: Series B (Statistical
  Methodology)}, 76(1):71--96.

\bibitem[Li et~al., 2020]{li2020deepkriging}
Li, Y., Sun, Y., and Reich, B.~J. (2020).
\newblock Deepkriging: Spatially dependent deep neural networks for spatial
  prediction.
\newblock {\em arXiv preprint arXiv:2007.11972}.

\bibitem[Lim et~al., 2017]{lim2017enhanced}
Lim, B., Son, S., Kim, H., Nah, S., and Mu~Lee, K. (2017).
\newblock Enhanced deep residual networks for single image super-resolution.
\newblock In {\em Proceedings of the IEEE conference on computer vision and
  pattern recognition workshops}, pages 136--144.

\bibitem[Reich and Fuentes, 2007]{reich2007multivariate}
Reich, B.~J. and Fuentes, M. (2007).
\newblock {A multivariate semiparametric Bayesian spatial modeling framework
  for hurricane surface wind fields}.
\newblock {\em {The Annals of Applied Statistics}}, 1(1):249--264.

\bibitem[Reich and Shaby, 2012]{reich2012hierarchical}
Reich, B.~J. and Shaby, B.~A. (2012).
\newblock A hierarchical max-stable spatial model for extreme precipitation.
\newblock {\em The Annals of Applied Statistics}, 6(4):1430.

\bibitem[Risser, 2016]{risser2016nonstationary}
Risser, M.~D. (2016).
\newblock {Nonstationary spatial modeling, with emphasis on process convolution
  and covariate--driven approaches}.
\newblock {\em arXiv preprint arXiv:1610.02447}.

\bibitem[Rodriguez and Dunson, 2011]{rodriguez2011nonparametric}
Rodriguez, A. and Dunson, D.~B. (2011).
\newblock {Nonparametric Bayesian models through probit stick-breaking
  processes}.
\newblock {\em {Bayesian Analysis}}, 6(1).

\bibitem[Romano et~al., 2019]{romano2019conformalized}
Romano, Y., Patterson, E., and Candes, E. (2019).
\newblock Conformalized quantile regression.
\newblock In {\em Advances in Neural Information Processing Systems}, pages
  3538--3548.

\bibitem[Shafer and Vovk, 2008]{shafer2008tutorial}
Shafer, G. and Vovk, V. (2008).
\newblock A tutorial on conformal prediction.
\newblock {\em Journal of Machine Learning Research}, 9(Mar):371--421.

\bibitem[Stein, 1990]{stein1990comparison}
Stein, M.~L. (1990).
\newblock A comparison of generalized cross validation and modified maximum
  likelihood for estimating the parameters of a stochastic process.
\newblock {\em The Annals of Statistics}, 18(3):1139--1157.

\bibitem[Stein, 2002]{stein2002screening}
Stein, M.~L. (2002).
\newblock The screening effect in kriging.
\newblock {\em The Annals of Statistics}, 30(1):298--323.

\bibitem[Tai et~al., 2017]{tai2017image}
Tai, Y., Yang, J., and Liu, X. (2017).
\newblock Image super-resolution via deep recursive residual network.
\newblock In {\em Proceedings of the IEEE conference on computer vision and
  pattern recognition}, pages 3147--3155.

\bibitem[Tang et~al., 2019]{tang2019identifiability}
Tang, W., Zhang, L., and Banerjee, S. (2019).
\newblock {On identifiability and consistency of the nugget in Gaussian spatial
  process models}.
\newblock {\em arXiv preprint arXiv:1908.05726}.

\bibitem[Tibshirani et~al., 2019]{conformal.shift}
Tibshirani, R.~J., Foygel~Barber, R., Candes, E., and Ramdas, A. (2019).
\newblock Conformal prediction under covariate shift.
\newblock In Wallach, H., Larochelle, H., Beygelzimer, A., d'Alch\'{e} Buc, F.,
  Fox, E., and Garnett, R., editors, {\em Advances in Neural Information
  Processing Systems 32}, pages 2530--2540. Curran Associates, Inc.

\bibitem[van~der Vaart, 1998]{vaart1998}
van~der Vaart, A.~W. (1998).
\newblock {\em Asymptotic Statistics}.
\newblock Cambridge University Press, Cambridge.

\bibitem[Vovk et~al., 2005]{vovk2005algorithmic}
Vovk, V., Gammerman, A., and Shafer, G. (2005).
\newblock {\em Algorithmic Learning in a Random World}.
\newblock Springer Science \& Business Media.

\bibitem[Wang et~al., 2019]{wang2019nearest}
Wang, H., Guan, Y., and Reich, B. (2019).
\newblock Nearest-neighbor neural networks for geostatistics.
\newblock In {\em 2019 International Conference on Data Mining Workshops
  (ICDMW)}, pages 196--205. IEEE.

\bibitem[Xu and Xie, 2020]{xu2020conformal}
Xu, C. and Xie, Y. (2020).
\newblock Conformal prediction for dynamic time-series.
\newblock {\em arXiv preprint arXiv:2010.09107}.

\bibitem[Zaffran et~al., 2022]{zaffran2022adaptive}
Zaffran, M., Dieuleveut, A., F{\'e}ron, O., Goude, Y., and Josse, J. (2022).
\newblock Adaptive conformal predictions for time series.
\newblock {\em arXiv preprint arXiv:2202.07282}.

\bibitem[Zhang, 2004]{zhang2004}
Zhang, H. (2004).
\newblock Inconsistent estimation and asymptotically equal interpolations in
  model-based geostatistics.
\newblock {\em J. Amer. Statist. Assoc.}, 99(465):250--261.

\end{thebibliography}

\end{document}